\documentclass[a4paper,12pt]{article}
\usepackage[a4paper]{geometry}
\usepackage[utf8]{inputenc}
\usepackage[T1]{fontenc}
\usepackage{color}
\usepackage{hyperref}
\usepackage{lscape}

\usepackage{amsmath,amssymb,amsthm}
\usepackage{thmtools}
\usepackage{stmaryrd}
\usepackage{array}
\usepackage{mathpartir}
\declaretheorem[numbered=yes,name=Lemma,qed=$\blacksquare$]{lemma}
\declaretheorem[numbered=yes,name=Theorem,qed=$\blacksquare$]{theorem}
\declaretheorem[numbered=yes,name=Definition,qed=$\blacksquare$]{definition}

\usepackage{tikz}
\usetikzlibrary{arrows}
\usetikzlibrary{positioning}
\usetikzlibrary{snakes}
\usetikzlibrary{decorations.pathreplacing}
\usepackage{mathtools}
\usepackage[makeroom]{cancel}
\usepackage{listings}
\usepackage{verbatim}
\usepackage{graphicx}
\usepackage{xcolor}
\usepackage{multirow}
\usepackage{enumitem}

\usepackage{natbib}
\newcommand{\STLC}{STLC}
\newcommand{\mulr}{$\mu$LR}

\newcommand{\var}[1]{\mathit{#1}}
\newcommand{\return}{\textit{Landin's}}

\newcommand{\dummy}{\var{dummy}}
\newcommand{\knot}{\var{knot}}

\newcommand{\domstyle}[1]{\mathrm{#1}}
\newcommand{\ClosedVal}{\domstyle{ClosedCal}}
\newcommand{\Type}{\domstyle{Type}}
\newcommand{\Loc}{\domstyle{Location}}
\newcommand{\Val}{\domstyle{Value}}
\newcommand{\UPred}[1]{\domstyle{UPred}(#1)}
\newcommand{\Pred}[1]{\domstyle{Pred(#1)}}
\newcommand{\World}{\domstyle{World}}
\newcommand{\blater}{\mathop{\blacktriangleright}}
\newcommand{\nats}{\mathbb{N}}

\newcommand{\typestyle}[1]{\mathit{#1}}
\newcommand{\bool}{\typestyle{bool}}
\newcommand{\integer}{\typestyle{int}}
\newcommand{\listt}{\typestyle{list}}
\renewcommand{\int}{\integer}
\newcommand{\tarray}{\typestyle{array}}
\newcommand{\tarrow}[2]{ #1 \rightarrow #2}
\newcommand{\tree}{\typestyle{tree}}
\newcommand{\Istack}{\typestyle{stack}}

\newcommand{\Tref}[1]{\typestyle{ref} \; #1}

\newcommand{\cut}[2][k]{\lfloor #2 \rfloor_{#1}}
\newcommand{\sem}[1]{\ensuremath{\llbracket #1 \rrbracket}}
\newcommand{\tuple}[1]{\ensuremath{\langle #1 \rangle}}
\newcommand{\curly}[1]{\mathcal{#1}}
\newcommand{\fun}{\rightarrow}
\newcommand{\finfp}{\xrightharpoonup{\textit{\tiny{fin}}}}
\newcommand{\monfp}{\xrightarrow{\textit{\tiny{mon}}}}

\DeclareMathOperator{\FTV}{FTV}
\DeclareMathOperator{\Rel}{Rel}
\DeclareMathOperator{\dom}{dom}
\DeclareMathOperator{\safe}{safe}
\DeclareMathOperator{\irred}{irred}
\DeclareMathOperator{\SNPred}{SN}
\DeclareMathOperator{\val}{Val}
\DeclareMathOperator{\inl}{inl}
\DeclareMathOperator{\inr}{inr}
\DeclareMathOperator{\scase}{case}
\DeclareMathOperator{\caseof}{of}
\DeclareMathOperator{\epack}{pack}
\DeclareMathOperator{\eunpack}{unpack}

\DeclareMathOperator{\ein}{in}
\DeclareMathOperator{\enot}{not}

\newcommand{\progstyle}[1]{\texttt{#1}}
\newcommand{\fold}{\progstyle{fold}}
\newcommand{\unfold}{\progstyle{unfold}}
\newcommand{\fst}{\progstyle{fst}}
\newcommand{\snd}{\progstyle{snd}}
\newcommand{\mk}{\progstyle{mk}}
\newcommand{\push}{\progstyle{push}}
\newcommand{\pop}{\progstyle{pop}}
\newcommand{\sort}{\progstyle{sort}}
\newcommand{\sortstring}{\progstyle{sortstring}}
\newcommand{\sortint}{\progstyle{sortint}}
\newcommand{\pack}[3]{\progstyle{pack} \; \tuple{#1,#2} \; \progstyle{as} \; #3}
\newcommand{\unpack}[4]{\progstyle{unpack} \; \tuple{#1,#2} = #3 \; \progstyle{in} \; #4 }

\newcommand{\Ealloc}[1]{\progstyle{ref} \; #1}
\newcommand{\Eassign}[2]{#1 \; \progstyle{:=} \; #2}
\newcommand{\Ederef}[1]{\progstyle{!} #1}

\newcommand{\evalto}{\rightarrow}
\newcommand{\evaltos}[1][*]{\evalto^{#1}}
\newcommand{\extendh}[3]{#1[#2\mapsto #3]}

\newcommand{\warning}[1]{{\color{red} !` #1 !} \\}
\newcommand{\mtenv}{{\raisebox{1.5pt}{$\scriptstyle\bullet$}}}
\newcommand{\case}[1]{\\{\bf Case} #1,}
\newcommand{\nequal}[1][k]{\stackrel{\tiny{#1}}{=}}
\newcommand{\eqdef}{\stackrel{\textit{\tiny{def}}}{=}}
\newcommand{\future}{\sqsupseteq}
\newcommand{\subst}[3]{#1[^{\textstyle #2}/_{\textstyle #3}]} %
\newcommand{\labs}[2]{\lambda #1 \ldotp #2}
\newcommand{\tlabs}[3]{\lambda #1 : #2 \ldotp \; #3}
\newcommand{\tLabs}[2][\alpha]{\Lambda #1 . \; #2}
\newcommand{\eif}[3]{\progstyle{if}\; #1 \; \progstyle{then} \; #2 \; \progstyle{else} \; #3}
\newcommand{\true}{\progstyle{true}}
\newcommand{\tstring}{\text{string}}
\newcommand{\false}{\progstyle{false}}
\newcommand{\SN}[2]{\SNPred_{#1}(#2)}

\newcommand{\pred}[2]{\curly{#1}\sem{#2}}
\newcommand{\pres}[3]{\ensuremath{\curly{#1}_{#2}\sem{#3}}}
\newcommand{\prep}[3]{\ensuremath{\curly{#1}\sem{#3}_{#2}}}

\newcommand{\epred}[1]{\pred{E}{#1}}
\newcommand{\epres}[2][k]{\pres{E}{#1}{#2}}
\newcommand{\eprep}[2][\rho]{\prep{E}{#1}{#2}}

\newcommand{\vpred}[1]{\pred{V}{#1}}
\newcommand{\vpres}[2][k]{\pres{V}{#1}{#2}}
\newcommand{\vprep}[2][\rho]{\prep{V}{#1}{#2}}

\newcommand{\gpred}[1]{\pred{G}{#1}}
\newcommand{\gpres}[2][k]{\pres{G}{#1}{#2}}
\newcommand{\gprep}[2][\rho]{\prep{G}{#1}{#2}}

\newcommand{\dpred}[1]{\pred{D}{#1}}

\newcommand{\hsat}[3][k]{#2 :_{#1} #3}

\newcommand{\sub}[3]{\subst{#1}{#2}{#3}}
\newcommand{\extsub}[3]{\ensuremath{#1 \lbrack #2 \mapsto #3 \rbrack}}

\newcommand{\equivalence}[3]{\ensuremath{#1 \approx^{#2} #3}}
\newcommand{\ctxeq}[2]{\equivalence{#1}{ctx}{#2}}
\newcommand{\lreq}[2]{\equivalence{#1}{LR}{#2}}

\newcommand{\TTrue}{\ensuremath{
    \inferrule*[right=T-True]{ }
               {\Gamma \vdash \true : \bool}}}

\newcommand{\TFalse}{\ensuremath{
    \inferrule*[right=T-False]{ }
               {\Gamma \vdash \false : \bool}}}
\newcommand{\TVar}{\ensuremath{
    \inferrule*[right=T-Var]{\Gamma(x) = \tau}
                            {\Gamma \vdash x : \tau}}}
\newcommand{\TIf}{\ensuremath{
    \inferrule*[right=T-If]{\Gamma \vdash e : \bool \and \Gamma \vdash e_1 : \tau \and \Gamma \vdash e_2 : \tau}
               {\Gamma \vdash \eif{e}{e_1}{e_2} : \tau}}}
\newcommand{\TApp}{\ensuremath{
    \inferrule*[right=T-App]{\Gamma \vdash e_1 : \tarrow{\tau_2}{\tau} \and
                            \Gamma \vdash e_2 : \tau_2}
                           {\Gamma \vdash e_1 \; e_2 : \tau}}}

\newcommand{\TAbs}{\ensuremath{\inferrule*[right=T-Abs]{\Gamma, x: \tau_1 \vdash e : \tau_2}
                           {\Gamma \vdash \tlabs{x}{\tau_1}{e} : \tarrow{\tau_1}{\tau_2}}}}
\newcommand{\TFold}{\ensuremath{
    \inferrule*[right=T-Fold]{\Gamma \vdash e : \tau[\mu\alpha. \; \tau/\alpha]}
                             {\Gamma \vdash \fold \; e : \mu\alpha. \; \tau}}}
\newcommand{\TUnfold}{\ensuremath{
    \inferrule*[right=T-Unfold]{\Gamma \vdash e : \mu\alpha. \; \tau}
                               {\Gamma \vdash \unfold \; e : \tau[\mu\alpha. \; \tau/\alpha]}}}

\newcommand{\FTTrue}{\ensuremath{
    \inferrule*[right=T-True]{ }
               {\Delta ; \Gamma \vdash \true : \bool}}}

\newcommand{\FTFalse}{\ensuremath{
    \inferrule*[right=T-False]{ }
               {\Delta ; \Gamma \vdash \false : \bool}}}

\newcommand{\FTVar}{\ensuremath{
    \inferrule*[right=T-Var]{\Gamma(x) = \tau}
                            {\Delta ; \Gamma \vdash x : \tau}}}
\newcommand{\FTApp}{\ensuremath{
    \inferrule*[right=T-App]{\Delta ; \Gamma \vdash e_1 : \tarrow{\tau_2}{\tau} \and
                            \Delta ; \Gamma \vdash e_2 : \tau_2}
                           {\Delta ; \Gamma \vdash e_1 \; e_2 : \tau}}}

\newcommand{\FTAbs}{\ensuremath{
    \inferrule*[right=T-Abs]{\Delta ; \Gamma, x: \tau_1 \vdash e : \tau_2}
                           {\Delta ; \Gamma \vdash \tlabs{x}{\tau_1}{e} : \tarrow{\tau_1}{\tau_2}}}}
\newcommand{\FTIf}{\ensuremath{
    \inferrule*[right=T-If]{\Delta ; \Gamma \vdash e : \bool \and 
                            \Delta ; \Gamma \vdash e_1 : \tau \\ 
                            \Delta ; \Gamma \vdash e_2 : \tau}
               {\Delta ; \Gamma \vdash \eif{e}{e_1}{e_2} : \tau}}}
\newcommand{\FTtApp}{\ensuremath{
    \inferrule*[right=T-TAbs]{\Delta,\alpha;\Gamma \vdash e : \tau}
                             {\Delta; \Gamma \vdash \tLabs{e} : \forall \alpha.\tau}}}
\newcommand{\FTtAbs}{\ensuremath{
    \inferrule*[right=T-TApp]{\Delta; \Gamma \vdash e : \forall \alpha . \tau \and
                              \Delta \vdash \tau'}
                             {\Delta ; \Gamma \vdash e [\tau'] : \subst{\tau}{\tau'}{\alpha}}}}

\author{Lau Skorstengaard\\lau@cs.au.dk}
\title{An Introduction to Logical Relations\\\large{Proving Program Properties Using Logical Relations}}
\begin{document}
\maketitle 
\tableofcontents
\section{Introduction}
\label{sec:introduction}
The term logical relations stems from Gordon Plotkin's memorandum \emph{Lambda-definability and logical relations written} in 1973.
However, the spirit of the proof method can be traced back to Wiliam W.
Tait who used it to show strong normalization of \emph{System T} in 1967.

%TODO insert actual references? http://homepages.inf.ed.ac.uk/gdp/publications/logical_relations_1973.pdf
%There was a discussion of the origin of the term Logical Relation and fundamental lemma on the types mailing list: http://lists.seas.upenn.edu/pipermail/types-list/2006/001054.html

Names are a curious thing.
When I say ``chair'', you immediately get a picture of a chair in your head.
If I say ``table'', then you picture a table.
The reason you do this is because we denote a chair by ``chair'' and a table by ``table'', but we might as well have said ``giraffe'' for chair and ``Buddha'' for table.
If we encounter a new word composed of known words, it is natural to try to find its meaning by composing the meaning of the components of the name.
Say we encounter the word ``tablecloth'' for the first time, then if we know what ``table'' and ``cloth'' denotes we can guess that it is a piece of cloth for a table.
However, this approach does not always work.
For instance, a ``skyscraper'' is not a scraper you use to scrape the sky.
Likewise for logical relations, it may be a fool's quest to try to find meaning in the name.
Logical relations are relations, so that part of the name makes sense.
They are also defined in a way that has a small resemblance to a logic, but trying to give meaning to logical relations only from the parts of the name will not help you understand them.
A more telling name might be Type Indexed Inductive Relations.
However, Logical Relations is a well-established name and easier to say, so we will stick with it (no one would accept ``giraffe'' to be a chair).

The majority of this note is based on the lectures of Amal Ahmed at the Oregon Programming Languages Summer School, 2015.
The videos of the lectures can be found at \url{https://www.cs.uoregon.edu/research/summerschool/summer15/curriculum.html}.

%\subsection{Organization} The note is structured as follows...

\subsection{Simply Typed Lambda Calculus}
The language we use to present logical predicates and relations is the simply typed lambda calculus (STLC).
In the first section, it will be used in its basic form, and later it will be used as the base language when we study new constructs and features.
We will later leave it implicit that it is \STLC{} that we extend with new constructs.
\STLC{} is defined in Figure~\ref{fig:stlc-def}.
\begin{figure}[htbp]
  \centering
  \[
    \arraycolsep=0pt\def\arraystretch{1.5}
    \begin{array}{p{1cm} r l }
      \multicolumn{3}{l}{\textbf{Types:}}\\
      &\tau ::={} & \bool \mid \tarrow{\tau}{\tau} \\
      \multicolumn{3}{l}{\textbf{Terms:}}\\
      &e ::={} &  x \mid \true \mid \false \mid \eif{e}{e}{e} \mid \tlabs{x}{\tau }{e} \mid e \; e\\
      \multicolumn{3}{l}{\textbf{Values:}}\\
      &v ::={} & \true \mid \false \mid \tlabs{x}{\tau}{e}\\
      \multicolumn{3}{l}{\textbf{Evaluation Context:}}\\
      &E ::= {} & [] \mid \eif{E}{e}{e} \mid E \; e \mid v \; E\\
      \multicolumn{3}{l}{\textbf{Evaluations:}} \\
      \multicolumn{3}{c}{
      \mathpar
      \inferrule{ }{ \eif{\true}{e_1}{e_2} \evalto e_1 } \and
      \inferrule{ }{ \eif{\false}{e_1}{e_2} \evalto e_2 } \and
      \inferrule{ }{(\tlabs{x}{\tau}{e}) \; v \evalto \subst{e}{v}{x}} \and
      \inferrule*[]{e \evalto e'}{E[e] \evalto E[e']}
      \endmathpar}\\
      \multicolumn{3}{l}{\textbf{Typing contexts:}} \\
      &\Gamma ::={} &  \mtenv \mid \Gamma , x : \tau\\
      \multicolumn{3}{l}{\textbf{Typing rules:}} \\
      \multicolumn{3}{c}{
      \mathpar
      \TFalse \and \TTrue \and
      \TVar \and \TAbs \and
      \TApp \and \TIf
      \endmathpar}
    \end{array}
  \]
  \caption{The simply typed lambda calculus. For the typing contexts, it is assumed that the binders ($x$) are distinct. That is, if $x \in \dom(\Gamma)$, then $\Gamma , x : \tau$ is not a legal context.}
  \label{fig:stlc-def}
\end{figure}

For readers unfamiliar with inference rules: A rule
\begin{mathpar}
  \inferrule{ A \and B }
            { A \wedge B }
\end{mathpar}
is read as if $A$ and $B$ is the case, then we can conclude $A \wedge B$. This means that the typing rule for application
\begin{mathpar}
  \TApp
\end{mathpar}
says that an application $e_1 \; e_2$ has the type $\tau$ under the typing context $\Gamma$ when $e_2$ has type $\tau_2$ under $\Gamma$ and $e_1$ has type $\tarrow{\tau_2}{\tau}$ also under $\Gamma$.

\subsection{Logical Relations}
\label{subsec:motivation-lr}
A logical relation is a proof method that can be used to prove properties of programs written in a particular programming language.
Proofs for properties of programming languages often go by induction on the typing or evaluation judgement.
A logical relations add a layer of indirection by constructing a collection of programs that all have the property we are interested in.
We will see this in more detail later.
As a motivation, here are a number of examples of properties that can be proven with a logical relation:
\begin{itemize}
\item Termination (Strong normalization)
\item Type safety
\item Program equivalences
  \begin{itemize}
  \item Correctness of programs
  \item Representation independence
  \item Parametricity and free theorems, e.g.
    \[
    f: \forall \alpha. \; \tarrow{\alpha}{\alpha}
    \]
    The program cannot inspect $\alpha$ as it has no idea which type it will be, therefore $f$ must be the identity function.
    \[
    \forall \alpha. \; \tarrow{\integer}{\alpha}
    \]
    A function with this type cannot exist (the function would need to return something of type $\alpha$, but it only has something of type $\integer$ to work with, so it cannot possibly return a value of the proper type).
  \item Security-Typed Languages (for Information Flow Control (IFC))\\
    Example: All types in the code snippet below are labelled with their security level.
    A type can be labelled with either $L$ for \emph{low} or $H$ for \emph{high}.
    We do not want any information flowing from variables with a \emph{high} label to a variable with a \emph{low} label.
    The following is an example of an insecure program because it has an \emph{explicit flow} of information from \emph{low} to \emph{high}:
        \begin{lstlisting}[escapeinside={@}{@}]
  x : int@$^L$@
  y : int@$^H$@
  x = y    //This assignment is insecure.
        \end{lstlisting}
    Information may also leak through a \emph{side channel}.
    There are many varieties of side channels, and they vary from language to language depending on their features.
    One of the perhaps simplest side channels is the following: Say we have two variable $x$ and $y$.
    They are both of integer type, but the former is labelled with \emph{low} and the latter with \emph{high}.
    Now say the value of $x$ depends on the value of $y$, e.g.\ $x=0$ when $y>0$ and $x=1$ otherwise.
    In this example, we may not learn the exact value of $y$ from $x$, but we will have learned whether $y$ is positive.
    The side channel we just sketched looks as follows:
        \begin{lstlisting}[escapeinside={@}{@}]
  x : int@$^L$@
  y : int@$^H$@
  if y > 0 then x = 0 else x = 1
        \end{lstlisting}
Generally, speaking the property we want is non-interference:
\begin{multline*}
  P(v_L,v_{1H}) \approx_L P(v_L,v_{2H})\\
  \text{for }\vdash P : \tarrow{\integer^L \times \integer^H}{\integer^L} 
\end{multline*}
That is for programs that generate \emph{low} results, we want ``\emph{low}-equivalent'' results.
\emph{Low}-equivalence means: if we execute $P$ twice with the same \emph{low} value but two different \emph{high} values, then the \emph{low} results of the two executions should be equal.
In other words, the execution cannot have depended on the \emph{high} value which means that no information was leaked to the \emph{low} results.
  \end{itemize}
\end{itemize}
\subsection{Categories of Logical Relations}
\label{subsec:categories-lr}
We can split logical relations into two: logical predicates and logical relations.
Logical predicates are unary and are usually used to show properties of programs.
Logical relations are binary and are usually used to show equivalences:
\begin{center}
  \begin{tabular}{l | l}
    Logical Predicates     & Logical Relations    \\
    \hline
    (Unary)                & (Binary)             \\
    $P_\tau(e)$             & $R_\tau(e_1,e_2)$     \\
    - One property         & - Program Equivalence\\ %\footnote{Was not in my notes.}
    - Strong normalization & \\
    - Type safety          & \\
  \end{tabular}
\end{center}
There are some properties that we want logical predicates and relation to have in general.
We describe these properties for logical predicates as they easily generalize to logical relations.
In general, we want the following things to hold true for a logical predicate that contains expressions $e$\footnote{Note: these are rules of thumb.
  For instance, one exception to the rule is the proof of type safety where the well-typedness condition is weakened to only require $e$ to be closed.}:
\begin{enumerate}
\item The expressions is closed and well-types, i.e.\ $\mtenv \vdash e : \tau$.
\item The expression has the property we are interested in.
\item The property of interest is preserved by eliminating forms.
\end{enumerate}

\section{Normalization of the Simply Typed Lambda Calculus}
\label{sec:stlc-strong-norm}
\subsection{Strong Normalization of \STLC{}}
In this section, we prove strong normalization for the simply typed lambda calculus which means that every term is strongly normalizing.
Normalization of a term is the process of reducing it to its normal form (where it can be reduced no further).
If a term is strongly normalizing, then it always reduces to its normal form.
In our case, the normal forms of the language are the values of the language.
\subsubsection*{A first attempt at proving strong normalization for \STLC{}}
We will first attempt a syntactic proof of the strong normalization property of \STLC{} to demonstrate how it fails.
However, first we need to state what we mean by strong normalization.
\begin{definition}
For expression $e$ and value $v$:
\begin{align*}
  e \Downarrow v & \eqdef e \evaltos v \\
  e \Downarrow   & \eqdef \exists v\ldotp e \Downarrow v \qedhere
\end{align*}
\end{definition}
\begin{theorem}[Strong Normalization] 
  If $\mtenv \vdash e : \tau$, then $e \Downarrow$
\end{theorem}
\begin{proof}\renewcommand{\qedsymbol}{$\boxtimes$}
\warning{This proof gets stuck and is not complete.}
Induction on the structure of the typing derivation.
\case{$\mtenv \vdash \true : \bool$} this term has already terminated.
\case{$\mtenv \vdash \false : \bool$} same as for \true.
\case{$\mtenv \vdash \eif{e}{e_1}{e_2} : \tau$} simple, but requires the use of canonical forms of bool\footnote{See Pierce's Types and Programming Languages~\citep{Pierce:types-and-pl} for more about canonical forms.}.
\case{$\mtenv \vdash \tlabs{x}{\tau_1}{e} : \tarrow{\tau_1}{\tau_2}$} it is a value already and it has terminated.
\case{$ \TApp $} \\
by the induction hypothesis, we get $e_1 \Downarrow v_1$ and $e_2 \Downarrow v_2$. By the type of $e_1$, we conclude $e_1 \Downarrow \tlabs{x}{\tau_2}{e'}$. What we need to show is $e_1 \; e_2 \Downarrow$. By the evaluation rules, we know $e_1 \; e_2$ takes the following steps:
\begin{align*}
  e_1 \; e_2 & \evaltos (\tlabs{x}{\tau_2}{e'}) \; e_2 \\
            & \evaltos (\tlabs{x}{\tau_2}{e'}) \; v_2 \\
            & \evalto e'[v_2/x]
\end{align*}
Here we run into an issue as we know nothing about $e'$.
As mentioned, we know from the induction hypothesis that $e_1$ evaluates to a lambda abstraction which makes $e_1$ strongly normalizing.
However, this say nothing about how the body of the body of the lambda abstraction evaluates.
Our induction hypothesis is simply not strong enough\footnote{:(}.
\end{proof}
The direct style proof did not work in this case, and it is not clear what to do to make it work.

\subsubsection*{Proof of strong normalization using a logical predicate}
Now that the direct proof failed, we try using a logical predicate. 
First, we define the predicate $\SN{\tau}{e}$:
\begin{align*}
  \SN{\bool}{e} & \Leftrightarrow{} \mtenv \vdash e : \bool \wedge e \Downarrow \\
  \SN{\tarrow{\tau_1}{\tau_2}}{e} & \Leftrightarrow{} \mtenv \vdash e : \tarrow{\tau_1}{\tau_2} \wedge e \Downarrow \wedge (\forall e'\ldotp  \SN{\tau_1}{e'} \implies \SN{\tau_2}{e \; e'})\\
\end{align*}
Now recall the three conditions from Section~\ref{subsec:categories-lr} that a logical predicate should satisfy.
It is easy to verify that $\SN{\tau}{e}$ only accepts closed well-typed terms.
Further, the predicate also requires terms to have the property we are interested in proving, namely $e \Downarrow$.
Finally, it should satisfy that ``\emph{the property of interest is preserved by eliminating forms}''. In \STLC{} lambdas are eliminated by application which means that application should preserve strong normalization when the argument is strongly normalizing.

The logical predicate is defined over the structure of $\tau$ which has $\bool$ as a base type, so the definition is well-founded\footnote{This may seem like a moot point as it is so obvious, but for some of the logical relations we see later it is not so, so we may as well start the habit of checking this now.}.
We are now ready to prove strong normalization using $\SN{\tau}{e}$.
To this end, we have the following lemmas:
\begin{lemma}
  \label{lem:pb-a}
    If $\mtenv \vdash e : \tau$, then $\SN{\tau}{e}$
\end{lemma}
\begin{lemma}
  \label{lem:pb-b}
If $\SN{\tau}{e}$, then $e \Downarrow$
\end{lemma}
These two lemmas are common for proofs using a logical predicate (or relation).
We first prove that all well-typed terms are in the predicate, and then we prove that all terms in the predicate have the property we want to show (in this case strong normalization).

The proof of Lemma~\ref{lem:pb-b} is by induction on $\tau$.
This proof is straightforward because the strong normalization was baked into the predicate.
It is generally a straightforward proof as our rules of thumb guide us to bake the property of interest into the predicate.

To prove Lemma~\ref{lem:pb-a}, we we could try induction over $\mtenv \vdash e : \tau$, but the case we will fail to show the case for \textsc{T-Abs}.
Instead we prove a generalization of Lemma~\ref{lem:pb-a}:
\begin{theorem}[Lemma \ref{lem:pb-a} generalized]
  \label{thm:sn-ftlr}
  If $\Gamma \vdash e : \tau$ and $\gamma \models \Gamma$, then $\SN{\tau}{\gamma(e)}$
\end{theorem}
This theorem uses a substitution $\gamma$ to close off the expression $e$.
In order for $\gamma$ to close off $e$, it must map all the possible variables of $e$ to strongly normalizing terms.
When we prove this lemma, we get a stronger induction hypothesis than we have when we try to prove Lemma~\ref{lem:pb-a}.
In Lemma~\ref{lem:pb-a}, the induction hypothesis can only be used with a closed term; but in this lemma, we can use an open term provided we have a substitution that closes it.

To be more specific, a substitution $\gamma = \{x_1 \mapsto v_1, \dots , x_n \mapsto v_n\}$ works as follows:
\begin{definition}
  \label{def:substitution}
  \begin{align*}
     \emptyset (e) ={}& e \\
     \extsub{\gamma}{x}{v}(e) ={}& \gamma(\subst{e}{v}{x}) \qedhere
  \end{align*}
\end{definition}
\footnote{We do not formally define substitution ($\subst{e}{v}{x}$). We refer to \citet{Pierce:types-and-pl} for a formal definition.}and $\gamma \models \Gamma$ is read ``the substitution $\gamma$ satisfies the type environment $\Gamma$'', and it is defined as follows:
\[
  \begin{gathered}
    \gamma \models \Gamma\\
    \text{iff}\\
    \dom(\gamma) = \dom(\Gamma) \wedge \forall x \in \dom(\Gamma)\ldotp    \SN{\Gamma(x)}{\gamma(x)}
  \end{gathered}
\]
To prove Theoremq~\ref{thm:sn-ftlr} we need two further lemmas
\begin{lemma}[Substitution Lemma] 
\label{lem:sn-subst}
  If $\Gamma \vdash e : \tau$ and $\gamma \models \Gamma$, then $\mtenv \vdash \gamma (e) : \tau$
\end{lemma}
\begin{proof}
  Left as an exercise.
\end{proof}
\begin{lemma}[$\SNPred$ preserved by forward/backward reduction]
  \label{lem:sn-preserved-by-red}
  Suppose $\mtenv \vdash e : \tau$ and $e \evalto e'$
  \begin{enumerate}
  \item if $\SN{\tau}{e'}$, then $\SN{\tau}{e}$
  \item if $\SN{\tau}{e}$, then $\SN{\tau}{e'}$
  \end{enumerate}
\end{lemma}
\begin{proof}
 Left as an exercise.
\end{proof}
\begin{proof}[Proof. (\ref{lem:pb-a} Generalized)] Proof by induction on $\Gamma \vdash e : \tau$.
\case{\textsc{T-True}} \\
Assume: 
\begin{itemize}
  \setlength\itemsep{0em}
  \item $\Gamma \vdash \true : \bool$
  \item $\gamma \models \Gamma$
\end{itemize}
We need to show:
\[
  \SN{\bool}{\gamma(\true)}
\]
There is no variable, so the substitution does nothing, and we just need to show $\SN{\bool}{\true}$ which is true as $\true \Downarrow \true$.
\case{\textsc{T-False}} similar to the \true{} case.
\case{\textsc{T-Var}}\\
Assume: 
\begin{itemize}
  \setlength\itemsep{0em}
  \item $\Gamma \vdash x : \tau$
  \item $\gamma \models \Gamma$
\end{itemize}
We need to show:
\[
  \SN{\tau}{\gamma(x)}
\]
This case follows from the definition of $\gamma \models \Gamma$.
We know that $x$ is well-typed, so it is in the domain of $\Gamma$.
From the definition of $\gamma \models \Gamma$, we get $\SN{\Gamma(x)}{\gamma(x)}$.
From well-typedness of $x$, we have $\Gamma(x) = \tau$ which then gives us what we needed to show.
\case{\textsc{T-If}} left as an exercise.
\case{\textsc{T-App}}\\
Assume:
\begin{itemize}
  \setlength\itemsep{0em}
  \item $\Gamma \vdash e_1 \; e_2 : \tau$
  \item $\gamma \models \Gamma$
\end{itemize}
We need to show:
\[
  \SN{\tau}{\gamma(e_1 \; e_2)}
\]
which amounts to $\SN{\tau}{\gamma(e_1) \; \gamma(e_2)}$.
By the induction hypothesis we have
\begin{align}
  &\SN{\tarrow{\tau_2}{\tau}}{\gamma(e_1)} \\
  &\SN{\tau_2}{\gamma(e_2)}
\end{align}
By the 3rd property of (1), $\forall e'\ldotp \SN{\tau_2}{e'} \implies \SN{\tau}{\gamma(e_1) \; e'}$, instantiated with (2), we get $\SN{\tau}{\gamma(e_1) \; \gamma(e_2)}$ which is the result we need.

Note this was the case we got stuck in when we tried to do the direct proof.
With the logical predicate, it is easily proven because we made sure to bake information about $e_1 \; e_2$ into $\SNPred_{\tarrow{\tau_2}{\tau}}$ when we follow the rule of thumb: ``\emph{The property of interest is preserved by eliminating forms}''. 
\case{\textsc{T-Abs}} \\
Assume: 
\begin{itemize}
  \setlength\itemsep{0em}
  \item $\Gamma \vdash \tlabs{x}{\tau_1}{e} : \tarrow{\tau_1}{\tau_2}$
  \item $\gamma \models \Gamma$
\end{itemize}
We need to show:
\[
  \SN{\tarrow{\tau_1}{\tau_2}}{\gamma(\tlabs{x}{\tau_1}{e})}
\]
which amounts to $\SN{\tarrow{\tau_1}{\tau_2}}{\tlabs{x}{\tau_1}{\gamma(e)}}$. Our induction hypothesis in this case reads:
\[
  \Gamma,x:\tau_1 \vdash e : \tau_2 \wedge \gamma' \models \Gamma, x : \tau_1 \quad \implies \quad \SN{\tau_2}{\gamma'(e)}
\]
It suffices to show the following three things:
\begin{enumerate}
\item \label{item:sn-abs-wt} $\mtenv \vdash \tlabs{x}{\tau_1}{\gamma(e)} : \tarrow{\tau_1}{\tau_2}$
\item \label{item:sn-abs-terminate} $\tlabs{x}{\tau_1}{\gamma(e)} \Downarrow$
\item \label{item:sn-abs-elim-pres} $\forall e'\ldotp \SN{\tau_1}{e'} \implies \SN{\tau_2}{(\tlabs{x}{\tau_1}{\gamma(e)}) \; e'}$
\end{enumerate}
If we use the substitution lemma (Lemma~\ref{lem:sn-subst}) and push the $\gamma$ in under the $\lambda$-abstraction, then we get \ref{item:sn-abs-wt}.
The lambda-abstraction is a value, so by definition \ref{item:sn-abs-terminate} is true.

It only remains to show \ref{item:sn-abs-elim-pres}.
To do this, we want to somehow apply the induction hypothesis for which we need a $\gamma'$ such that $\gamma' \models \Gamma, x:\tau_1$.
We already have $\gamma$ and $\gamma \models \Gamma$, so our $\gamma'$ should probably have the form $\gamma' = \gamma[x \mapsto v_?]$ for some $v_?$ of type $\tau_1$.
Let us move on and see if any good candidates for $v_?$ present themselves.

Let $e'$ be given and assume $\SN{\tau_1}{e'}$.
We then need to show $\SN{\tau_2}{(\tlabs{x}{\tau_1}{\gamma(e)}) \; e'}$.
From $\SN{\tau_1}{e'}$, it follows that $e' \Downarrow v'$ for some $v'$.
$v'$ is a good candidate for $v_?$ so let $v_?
= v'$.
From the forward part of the preservation lemma (Lemma~\ref{lem:sn-preserved-by-red}), we can further conclude $\SN{\tau_1}{v'}$.
We use this to conclude $\gamma[x\mapsto v'] \models \Gamma, x:\tau_1$ which we use with the assumption $\Gamma,x:\tau_1 \vdash e : \tau_2$ to instantiate the induction hypothesis and get $\SN{\tau_2}{\gamma[x\mapsto v'](e)}$.

Now consider the following evaluation:
\begin{align*}
  (\tlabs{x}{\tau_1}{\gamma(e)}) \; e' & \evaltos (\tlabs{x}{\tau_1}{\gamma(e)}) \; v' \\
                                       & \evalto \gamma(e)[v'/x] \equiv 
                                                   \gamma[x \mapsto v'](e)
\end{align*}
We already concluded that $e' \evaltos v'$, which corresponds to the first series of steps.
We can then do a $\beta$-reduction to take the next step, and finally we get something that is equivalent to $\gamma[x \mapsto v'](e)$.
That is we have the evaluation
\[
(\tlabs{x}{\tau_1}{\gamma(e)}) \; e' \evaltos \gamma[x \mapsto v'](e)
\]
From $\SN{\tau_1}{e'}$, we have $\mtenv \vdash e' : \tau_1$ and we already argued that $\mtenv \vdash \tlabs{x}{\tau_1}{\gamma(e)} : \tarrow{\tau_1}{\tau_2}$, so from the application typing rule we get $\mtenv \vdash (\tlabs{x}{\tau_1}{\gamma(e)}) \; e' : \tau_2$.
We can use this with the above evaluation and the forward part of the preservation lemma (Lemma~\ref{lem:sn-preserved-by-red}) to argue that every intermediate expressions in the steps down to $\gamma[x \mapsto v'](e)$ are closed and well typed.

If we use $\SN{\tau_2}{\gamma[x\mapsto v'](e)}$ with $(\tlabs{x}{\tau_1}{\gamma(e)}) \; e' \evaltos \gamma[x \mapsto v'](e)$ and the fact that every intermediate step in the evaluation is closed and well typed, then we can use the backward reduction part of the $\SNPred$ preservation lemma to get $\SN{\tau_2}{(\tlabs{x}{\tau_1}{\gamma(e)}) \; e'}$ which is the result we wanted.
\end{proof}
\subsection{Exercises}\label{sec:SN:exercises}
\begin{enumerate}
\item Prove $\SNPred$ preserved by forward/backward reduction (Lemma~\ref{lem:sn-preserved-by-red}).
\item Prove the substitution lemma (Lemma~\ref{lem:sn-subst}).
\item Go through the cases of the proof for Theorem~\ref{thm:sn-ftlr} by yourself.
\item Prove the \textsc{T-If} case of Theorem~\ref{thm:sn-ftlr}.
\item Extend the language with pairs and adjust the proofs. \label{ex:extend-STLC-pairs}

  Specifically, how do you apply the rules of thumb for the case of pairs?
  Do you need to add anything for the third clause (eliminating forms preservation), or does it work without doing anything for it like it did for case of booleans?
\end{enumerate}

recu\section{Type Safety for \STLC{}}
\label{sec:stlc-type-safety}
In this section, we prove type safety for simply typed lambda calculus using a logical predicate.

First we need to consider what type safety is.
The classical mantra for type safety is: ``Well-typed programs do not \emph{go wrong}.''
It depends on the language and type system what \emph{go wrong} actually means, but in our case a program has \emph{gone wrong} when it is stuck\footnote{In the case of language-based security and information flow control, the notion of \emph{go wrong} would be that there is an undesired flow of information.} (an expression is stuck if it is irreducible but not a value).
\subsection{Type safety - the classical treatment}
Type safety for simply typed lambda calculus is stated as follows:
\begin{theorem}[Type Safety for STLC]
  If $\mtenv \vdash e : \tau$ and $e \evaltos e'$, then $\val(e')$ or $\exists e''\ldotp e' \evalto e''$.
\end{theorem}
Traditionally, the type safety proof uses two lemmas: progress and preservation.
\begin{lemma}[Progress]
  If $\mtenv \vdash e : \tau$, then $\val(e)$ or $\exists e'\ldotp \evalto e'$.
\end{lemma}
Progress is normally proved by induction on the typing derivation.
\begin{lemma}[Preservation]
  If $\mtenv \vdash e : \tau$ and $e \evalto e'$, then $\mtenv \vdash e' : \tau$.
\end{lemma}
Preservation is normally proved by induction on the evaluation.
Preservation is also known as \emph{subject reduction}.
Progress and preservation talk about one step, so to prove type safety we have to do induction on the evaluation.
Here we do not want to prove type safety the traditional way, but if you are unfamiliar with it and want to learn more, then we refer to Pierce's \emph{Types and Programming Languages}~\citep{Pierce:types-and-pl}.

We will use a logical predicate (as it is a unary property) to prove type safety.

\subsection{Type safety - using logical predicate}
We define the predicate in a slightly different way compared to Section~\ref{sec:stlc-strong-norm}.
We define it in two parts: a value interpretation and an expression interpretation.
The value interpretation is a function from types to the power set of closed values:
\[
  \vpred{-} : \tarrow{\Type}{\curly{P}(\ClosedVal)}
\]
The value interpretation is defined as:
\begin{align*}
  \vpred{\bool} & = \{ \true, \false \}\\ 
  \vpred{\tarrow{\tau_1}{\tau_2}} & = \{\tlabs{x}{\tau_1}{e} \mid \forall v \in \curly{V}\sem{\tau_1}\ldotp e [v/x] \in \curly{E}\sem{\tau_2}\}
\end{align*}
We define the expression interpretation as:
\[
  \epred{\tau} = \{e \mid \forall e'\ldotp e \evaltos e' \wedge \irred(e') \implies e' \in \curly{V}\sem{\tau} \}
\]
Notice that neither $\vpred{\tau}$ nor $\epred{\tau}$ requires well-typedness.
Normally, the logical predicate would require this as our guidelines suggest it.
However, as the goal is to prove type safety we do not want it as a part of the predicate.
In fact, if we did include a well-typedness requirement, then we would end up proving preservation for some of the proofs to go through.
We do, however, require the value interpretation to only contain closed values.

An expression is irreducible if it is unable to take any reduction steps according to the evaluation rules.
The predicate $\irred$ captures whether an expression is irreducible:
\[
  \begin{gathered}
    \irred(e)\\
    \text{iff}\\
    \nexists e'\ldotp e \evalto e'
  \end{gathered}
\]
The sets are defined on the structure of the types. $\vpred{\tarrow{\tau_1}{\tau_2}}$ contains $\epred{\tau_2}$, but $\epred{\tau_2}$ uses $\tau_2$ directly in $\vpred{\tau_2}$, so the definition is structurally well-founded. To prove type safety, we first define a new predicate, $\safe$:
\[
  \begin{gathered}
    \safe(e)\\
    \text{iff}\\
    \forall e' \ldotp e \evaltos e' \implies \val(e') \vee \exists e''\ldotp e' \evalto e''
  \end{gathered}
\] %TODO Check this argument! It was evaltos in the last evaluation, but that did not make sense!
An expression $e$ is safe if it can take a number of steps and end up either as a value or as an expression that can take another step.

We are now ready to prove type safety.
Just like we did for strong normalization, we use two lemmas:
%TODO: Find something better for this. % Ask Kristoffer.
\begin{lemma}
  \label{lem:stlc-ts-weak-ftlr}
  If $\mtenv \vdash e : \tau$, then $e \in \epred{\tau}$
\end{lemma}
\begin{lemma}
  \label{lem:stlc-ts-wrong-compat}
  If $e \in \epred{\tau}$, then $\safe(e)$
\end{lemma}
Rather than proving Lemma~\ref{lem:stlc-ts-weak-ftlr} directly, we prove a more general theorem and get Lemma~\ref{lem:stlc-ts-weak-ftlr} as a consequence.
We are not yet in a position to state the generalization.
First, we need to define the interpretation of environments:
\begin{align*}
  \gpred{\mtenv} & = \{ \emptyset \} \\
  \gpred{\Gamma,x:\tau} & = \{\gamma[x \mapsto v] \mid
    \gamma \in \gpred{\Gamma} \wedge 
    v \in \vpred{\tau}\}
\end{align*}
Further, we need to define semantic type safety:
\[
  \begin{gathered}
    \Gamma \models e : \tau\\
    \text{iff}\\
    \forall \gamma \in \gpred{\Gamma} \ldotp \gamma(e) \in \epred{\tau}
  \end{gathered}
\]
This definition should look familiar because we use the same trick as we did for strong normalization: Instead of only considering closed terms, we consider all terms but require a substitution that closes it.

We can now define our generalized version of Lemma~\ref{lem:stlc-ts-weak-ftlr}:
\begin{theorem}[Fundamental Property]
  \label{thm:stlc-ts-ftlr}
  If $\Gamma \vdash e : \tau$, then $\Gamma \models e : \tau$
\end{theorem}
A theorem like this would typically be the first you prove after defining a logical relation.
In this case, the theorem says that syntactic type safety implies semantic type safety.

We also alter Lemma~\ref{lem:stlc-ts-wrong-compat} to fit with Theorem~\ref{thm:stlc-ts-ftlr}:
\begin{lemma}
  \label{lem:stlc-ts-compat}
  If $\emptyset \models e : \tau$, then $\safe(e)$
\end{lemma}
\begin{proof}
  Suppose $e \evaltos e'$ for some $e'$, then we need to show $\val(e')$ or $\exists e''\ldotp e' \evalto e''$.
  We proceed by case on whether or not $\irred(e')$:
  \case{$\neg \irred(e')$} this case follows directly from the definition of $\irred$:
  $\irred(e')$ is defined as $\nexists e''\ldotp e' \evalto e''$ and as the assumption is $\neg \irred(e')$, we get $\exists e''\ldotp e' \evalto e''$.
  \case{$\irred(e')$} by assumption we have $\mtenv \models e : \tau$.
  As the typing context is empty, we choose the empty substitution and get $e \in \epred{\tau}$.
  We now use the definition of $e \in \epred{\tau}$ with the two assumptions $e \evaltos e'$ and $\irred(e')$ to conclude $e' \in \vpred{\tau}$.
  As $e'$ is in the value interpretation of $\tau$, we can conclude $\val(e')$.
\end{proof}
To prove the Fundamental Property (Theorem~\ref{thm:stlc-ts-ftlr}), we need a substitution lemma:
\begin{lemma}[Substitution]
Let $e$ be syntactically well-formed term, let $v$ be a closed value and let $\gamma$ be a substitution that maps term variables to closed values, and let $x$ be a variable not in the domain of $\gamma$, then
\[
\extsub{\gamma}{x}{v}(e) = \subst{\gamma(e)}{v}{x} \qedhere
\]
\end{lemma}
\begin{proof}
  By induction on the size of $\gamma$.
  \case{$\gamma = \emptyset$} this case is immediate by definition of substitution. That is by definition we have $[x \mapsto v] e = \subst{e}{v}{x}$.
  \case{$\gamma = \gamma'[y \mapsto v']$, $x\neq y$} in this case our induction hypothesis is:
\[
  \gamma'[x \mapsto v] e = \subst{\gamma'(e)}{v}{x}
\]
We wish to show 
\[
  \gamma'[y \mapsto v'][x \mapsto v] e = \subst{\gamma'[y \mapsto v'](e)}{v}{x}
\]
\begin{align}
    \gamma'[y \mapsto v'][x \mapsto v] e & = \gamma'[x \mapsto v][y \mapsto v'] e \label{sub:step1}\\
                                         & = \gamma'[x \mapsto v] (\subst{e}{v'}{y}) \label{sub:step2}\\
                                         & = \subst{\gamma'(\subst{e}{v'}{y})}{v}{x}\label{sub:step3}\\
                                         & = \subst{\gamma'[y \mapsto v'](e)}{v}{x}\label{sub:step4}
\end{align}
In the first step (\ref{sub:step1}), we swap the two mappings.
It is safe to do so as both $v$ and $v'$ are closed values, so we know that no variable capturing will occur.
In the second step (\ref{sub:step2}), we just use the definition of substitution (Definition~\ref{def:substitution}).
In the third step (\ref{sub:step3}), we use the induction hypothesis\footnote{The induction hypothesis actually has a number of premises, as an exercise convince yourself that they are satisfied.}.
Finally in the last step (\ref{sub:step4}), we use the definition of substitution to get the $y$ binding out as an extension of $\gamma'$.
\end{proof}
\begin{proof}[Proof. (Fundamental Property, Theorem~\ref{thm:stlc-ts-ftlr})] Proof by induction on the typing judgement.
  \case{\textsc{T-Abs}} \\
assuming $\Gamma \vdash \tlabs{x}{\tau_1}{e} : \tarrow{\tau_1}{\tau_2}$, we need to show $\Gamma \models \tlabs{x}{\tau_1}{e} : \tarrow{\tau_1}{\tau_2}$. Suppose $\gamma \in \gpred{\Gamma}$ and show
\[
  \gamma(\tlabs{x}{\tau_1}{e}) \in \epred{\tarrow{\tau_1}{\tau_2}} \equiv
  (\tlabs{x}{\tau_1}{\gamma(e)}) \in \epred{\tarrow{\tau_1}{\tau_2}}
\]
Now suppose that $\tlabs{x}{\tau_1}{\gamma(e)} \evaltos e'$ and $\irred(e')$.
We then need to show $e' \in \vpred{\tarrow{\tau_1}{\tau_2}}$.
$\tlabs{x}{\tau_1}{\gamma(e)}$ is irreducible because it is a value, and we can conclude it takes no steps.
In other words $e' = \tlabs{x}{\tau_1}{\gamma(e)}$.
This means we need to show $\tlabs{x}{\tau_1}{\gamma(e)} \in \vpred{\tarrow{\tau_1}{\tau_2}}$.
Now suppose $v \in \vpred{\tau_1}$, then we need to show $\gamma(e)[v/x] \in \epred{\tau_2}$.

Keep the above proof goal in mind and consider the induction hypothesis:
\[
  \Gamma, x: \tau_1 \models e : \tau_2
\]
Instantiate this with $\gamma[x \mapsto v]$.
We have $\gamma[x \mapsto v] \in \gpred{\Gamma, x : \tau_1}$ because of assumptions $\gamma \in \gpred{\Gamma}$ and $v \in \vpred{\tau_2}$.
The instantiation gives us $\gamma[x \mapsto v] (e) \in \epred{\tau_2} \equiv \gamma(e)[v/x] \in \epred{\tau_2}$.
The equivalence is justified by the substitution lemma we proved.
This is exactly the proof goal we kept in mind.
\case{\textsc{T-App}} Show this case as an exercise.\\
The remaining cases are straightforward.
\end{proof}
Now consider what happens if we add pairs to the language (exercise \ref{ex:extend-STLC-pairs} in exercise section~\ref{sec:SN:exercises}).
% I decided not to add a clarify of what additions to STLC are needed for pairs. In case I change my below are some
\begin{comment}
\begin{align*}
  &\fst \tuple{v_1,v_2} \evalto v_1 \\
  &\snd \tuple{v_1,v_2} \evalto v_2
\end{align*}
\end{comment}
We need to add a clause to the value interpretation:
\[
  \vpred{\tau_1 \times \tau_2} = \{\tuple{v_1,v_2} \mid v_1 \in \vpred{\tau_1} \wedge v_2 \in \vpred{\tau_2}\}
\]
There is nothing surprising in this addition to the value relation, and it should not be a challenge to show the pair case of the proofs.
%omitted: '3rd part of LR recipe is mostly about functions - Contravariance problem.'

If we extend our language with sum types. %that is
\[
e::= \dots \mid \inl \; v \mid \inr \; v \mid \scase \; e \caseof \; \inl \; x => e_1 \quad \inr \; x => e_2
\]
Then we need to add the following clause to the value interpretation:
\[
  \vpred{\tau_1 + \tau_2} = \{\inl \; v \mid v \in \vpred{\tau_1}\} \cup
                           \{\inr \; v \mid v \in \vpred{\tau_2}\}
\]
It turns out this clause is sufficient.
One might think that it is necessary to require the body of the match to be in the expression interpretation, which looks something like $\forall e \in \epred{\tau}$.
This requirement will, however, give well-foundedness problems, as $\tau$ is not a structurally smaller type than $\tau_1 + \tau_2$.
It may come as a surprise that we do not need to relate the expressions as the slogan for logical relations is: ``Related inputs to related outputs.''
\subsection{Exercises}
\begin{enumerate}
\item Prove the \textsc{T-App} case of the Fundamental Property (Theorem~\ref{thm:stlc-ts-ftlr}).
\item Verify the remaining cases of Theorem~\ref{thm:stlc-ts-ftlr} (\textsc{T-True},\textsc{T-False},\textsc{T-Var}, and \textsc{T-If}).
\end{enumerate}

\section{Universal Types and Relational Substitutions}
In the previous sections, we considered the unary properties and safety and termination, but now we shift our focus to relational properties and specifically program equivalences.
Generally speaking, we use logical relations rather than predicates for proving relational properties.
A program equivalence is a relational property, so we are going to need a logical relation.

We will consider the language System F which is \STLC{} with universal types.
Universal types allow us to write generic functions.
For instance, say we have a function that sorts integer lists:
\[
  \sortint \; : \; \tarrow{\listt \; \integer}{\listt \; \integer}
\]
The function takes a list of integers and returns the sorted list.
Say we now want a function $\sortstring$ that sorts lists of strings.
Instead of implementing a completely new sorting function for this, we could factor out the generic code responsible for sorting from $\sortint$ and make a generic function.
The type signature for a generic sort function would be:
\[
  \sort \; : \; \forall \alpha\ldotp \tarrow{(\listt \; \alpha) \times (\tarrow{\alpha \times \alpha}{bool} )}{\listt \; \alpha}
\]
The generic sort function takes a type, a list of elements of this type, and a comparison function that compares to elements of the given type.
The result of the sorting function is a list sorted according to the comparison function.
An example of an application of this function could be
\[
  \sort \; [\int] \; (3,7,5) \; \texttt{<}
\]
Whereas sort instantiated with the $\tstring$ type, but given an integer list would not be a well typed instantiation.
\[
  \sort \; [\tstring] \; \cancelto{("a","c","b")}{\color{red}(3,7,5)} \; \texttt{string<}
\]
Here the application with the list $(3,7,5)$ is not well typed, but if we instead use a list of strings, then it type checks. 

We want to extend the simply typed lambda calculus with functions that abstract over types in the same way lambda abstractions, $\tlabs{x}{\tau}{e}$, abstract over terms. We do that by introducing a type abstraction:
\[
  \tLabs{e}
\]
This function abstracts over the type $\alpha$ which allows $e$ to depend on $\alpha$.

\subsection{System F (STLC with universal types)}
We extend \STLC{} as follows to get System F:
\[
\arraycolsep=0pt
\begin{array}{r l}
  \tau ::={} & \ldots \mid \forall \alpha\ldotp \tau \\
  e    ::={} & \ldots \mid \tLabs{e} \mid e[\tau] \\
  v    ::={} & \ldots \mid \tLabs{e}\\
  E    ::={} & \ldots \mid E[\tau] 
\end{array}
\]
New evaluation:
\begin{mathpar}
  \inferrule{ }{(\tLabs{e})[\tau] \evalto e[\tau/\alpha]}
\end{mathpar}
Type environment:
\[
  \Delta ::= \mtenv \mid \Delta,\; \alpha
\]
The type environment consists of type variables, and they are assumed to be distinct type variables.
That is, the environment $\Delta,\alpha$ is only well-formed if $\alpha \not\in \dom(\Delta)$\footnote{We do not annotate $\alpha$ with a kind as we only have one kind in this language.}.
With the addition of type environments, we update the form of the typing judgement as follows
\[
  \Delta,\Gamma \vdash e : \tau
\]
We now need a notion of well-formed types. If $\tau$ is well formed with respect to $\Delta$, then we write:
\[
  \Delta \vdash \tau
\]
We do not include the formal rules here, but they amount to $\FTV(\tau) \subseteq \Delta$, where $\FTV(\tau)$ is the set of free type variables in $\tau$.

We further introduce a notion of well-formed typing contexts. A context is well formed if all the types that appear in the range of $\Gamma$ are well formed.
\[
  \begin{gathered}
    \Delta \vdash \Gamma\\
    \text{iff}\\
    \forall x \in \dom(\Gamma)\ldotp \Delta \vdash \Gamma(x)
\end{gathered}
\]
For any type judgment $\Delta,\Gamma \vdash e : \tau$, we have as an invariant that $\tau$ is well formed in $\Delta$ and $\Gamma$ is well formed in $\Delta$.
The old typing system modified to use the new form of the typing judgment looks like this:

\begin{mathpar}
  \FTFalse
  \and
  \FTTrue
  \and
  \FTVar
  \and
  \FTIf
  \and
  \FTAbs 
  \and
  \FTApp
\end{mathpar}
Notice that the only thing that has changed is that $\Delta$ has been added to the environment in the judgments.
We further extend the typing rules with the following two rules to account for our new language constructs:
\begin{mathpar}
  \FTtAbs
  \and
  \FTtApp
\end{mathpar}

\subsection{Properties of System-F: Free Theorems}
\label{subsec:free-thm}
In System-F, certain types reveal the behavior of the functions with that type.
Let us consider terms with the type $\forall \alpha\ldotp \tarrow{\alpha}{\alpha}$.
Recall from Section~\ref{subsec:motivation-lr} that this has to be the identity function.
We can now phrase this as a theorem:
\begin{theorem}
  \label{thm:free-thm-identity}
  If all of the following hold
  \begin{itemize}  \setlength\itemsep{0em} \renewcommand\labelitemi{--}
  \item $\mtenv ; \mtenv \vdash e : \forall \alpha\ldotp \tarrow{\alpha}{\alpha}$
  \item $\mtenv \vdash \tau$
  \item $\mtenv; \mtenv \vdash v : \tau$,
  \end{itemize}
then
\[
  e[\tau]\; v \evaltos v
  \qedhere
\]
\end{theorem}
This is a free theorem in this language.
Another free theorem from Section~\ref{subsec:motivation-lr} is for expressions of type $\forall \alpha\ldotp \tarrow{\alpha}{\mathit{bool}}$.
All expressions with this type has to be constant functions.
We can also phrase this as a theorem:
\begin{theorem}
  If all of the following hold
  \begin{itemize}\setlength\itemsep{0em} \renewcommand\labelitemi{--}
  \item $\mtenv ; \mtenv \vdash e : \forall \alpha\ldotp \tarrow{\alpha}{\mathit{bool}}$
  \item $\mtenv \vdash \tau$
  \item $\mtenv ; \mtenv \vdash v_1 : \tau$
  \item $\mtenv ; \mtenv \vdash v_2 : \tau$
\end{itemize}
then
\[
  \ctxeq{e[\tau] \; v_1}{e[\tau]\; v_2}
  \qedhere
\]
\end{theorem}
\footnote{We have not defined $\ctxeq{}{}$ yet.
  For now, it suffices to know that it equates programs that are behaviorally the same.}
We can even allow the type abstraction to be instantiated with different types:
\begin{theorem}
  \label{thm:free-thm-three}
  If all of the following hold
  \begin{itemize}\setlength\itemsep{0em} \renewcommand\labelitemi{--}
  \item $\mtenv ; \mtenv \vdash e : \forall \alpha\ldotp \tarrow{\alpha}{\mathit{bool}}$
  \item $\mtenv \vdash \tau_1$
  \item $\mtenv \vdash \tau_2$
  \item $\mtenv ; \mtenv \vdash v_1 : \tau_1$
  \item $\mtenv ; \mtenv \vdash v_2 : \tau_2$
  \end{itemize}
  then
  \[
    \ctxeq{e[\tau_1]\; v_1}{e[\tau_2]\; v_2}
    \qedhere
  \]
\end{theorem}
We get these free theorems because the functions do not know the type of the argument which means that they have no way to inspect it.
The function can only treat its argument as an unknown ``blob'', so it has no choice but to return the same value every time.

The question now is: how do we prove these free theorems?
The two last theorems both talk about program equivalence, so the proof technique of choice is a logical relation.
The first theorem did not mention a program equivalence, but it can also be proven with a logical relation.

\subsection{Contextual Equivalence}
To define a contextual equivalence, we first define the notion of a program context.
A program context is a complete program with exactly one hole in it:
\[
  \arraycolsep=0pt
  \begin{array}{rl}
    C ::={} &  [\cdot] \mid \eif{C}{e}{e}  \mid \eif{e}{C}{e} \mid \eif{e}{e}{C} \mid \\
            & \tlabs{x}{\tau}{C} \mid C \; e \mid e \; C \mid \tLabs{C} \mid C[\tau]
  \end{array}
\]
For instance, \[\tlabs{y}{\bool}{[\cdot]}\] is a context where the hole is the body of the lambda abstraction.

We need a notion of context typing. For simplicity, we just introduce it for simply typed lambda calculus. The context typing is written as:
\[
\inferrule*{\Gamma \vdash e : \tau \and
            \Gamma' \vdash C[e] : \tau'}
           {C\;:\; (\Gamma \vdash \tau) \implies (\Gamma' \vdash \tau')} 
\]
This means that for any expression $e$ of type $\tau$ under $\Gamma$ if we embed it into $C$, then the type of the embedding is $\tau'$ under $\Gamma'$.
For our example of a context, we would have
\[
  \tlabs{y}{\bool}{[\cdot]}\;:\; (y : \bool \vdash \bool) \implies (\mtenv \vdash \tarrow{\bool}{\bool})
\]
because the hole of the context can be plugged with any expression of type $\bool$ with variable $y$ free in it.
Such an expression could be $\eif{y}{\false}{\true}$ as $y : \bool \vdash \eif{y}{\false}{\true} : \bool$.
When the hole in the context is plugged with such an expression, then it is closed and has type $\tarrow{\bool}{\bool}$.
For instance, if we plug the example context with the above expression we have
\[ 
  \mtenv \vdash \tlabs{y}{\bool}{\eif{y}{\false}{\true}} : \tarrow{\bool}{\bool}
\]

Informally we want contextual equivalence to express that two expressions give the same result no matter what program context they are plugged into.
In other words, two expressions are contextually equivalent when no program context is able to observe any difference in behavior between the two expressions.
For this reason, contextual equivalence is also known as \emph{observational equivalence}.
A hole has to be plugged with a term of the correct type, so we annotate the equivalence with the type of the hole which means that the two contextually equivalent expressions must have that type.
\[
\begin{gathered}
  \Delta; \Gamma \vdash \ctxeq{e_1}{e_2} : \tau \\
  \text{iff} \\
  \forall C\;:\; (\Delta ; \Gamma \vdash \tau ) \implies (\mtenv ; \mtenv \vdash bool)\ldotp (C[e_1] \Downarrow v \iff C[e_2] \Downarrow v)
\end{gathered}
\]
This definition assumes that $e_1$ and $e_2$ have type $\tau$ under the specified contexts.
Notice from the context typing we have that $C$ plugged with any term must be closed which ensures that it can be evaluated.
Further, we require the plugged context to have type $\bool$.
If we allowed the plugged context to be of function type, then the two functions produced by $C[e_1]$ and $C[e_2]$ would have to be equivalent.
This would of course also need to hold true for the program context that is just a hole because we need to decide whether $e_1$ and $e_2$ behaves the same which is exactly what we try to define with contextual equivalence.
\footnote{In non-terminating languages, the definition of contextual equivalence only mentions termination and not whether the values produced are equal.
  For non-terminating languages this makes sense because if a context is able to differentiate two expressions, then it is easy to construct another context that terminates when plugged with one expression but diverges when plugged with the other.
}.

Contextual equivalence can be used to argue about correctness of a program.
For instance, say we have two implementations of a stack, one is implemented using an array and the other using a list.
If we can show that the two implementations are contextual equivalent, then we can use either stack implementation in any context with the same result.
In particular, if one of the stack implementations is highly optimized, then we would like to use that implementation.
However, it is not always clear whether such highly optimized programs do what they are supposed to.
In order to argue this, one could implement a stack reference or specification implementation which is slower but where it is easy to see that it behaves like a stack simply by inspecting the implementation.
A proof of contextual equivalence between the highly optimized stack implementation and the specification implementation could be seen as a correctness proof \emph{with respect to} the specification.

In the next part, we will introduce a logical relation such that
\[
  \Delta ; \Gamma \vdash \lreq{e_1}{e_2} : \tau \implies  \Delta ; \Gamma \vdash \ctxeq{e_1}{e_2} : \tau
\]
That is we want to show that the logical relation is sound with respect to contextual equivalence. 

If we can prove the above soundness, then we can state our free theorems with respect to $\lreq{}{}$ and get the result with respect to $\ctxeq{}{}$ as a corollary. 
In fact, it is very difficult to prove contextual equivalences directly as the proof has to consider every possible program context.
Such a proof would be by induction, but the lambda-abstraction case would be very difficult.

\subsection{A Logical Relation for System F}
In this section, we construct the logical relation for System F, so we can prove the free theorems of Section~\ref{subsec:free-thm}.
We are defining a relation, so the value interpretation consists of pairs.
All pairs in the value relation must be closed and well-typed:
\[
  \vpred{\tau} = \{(v_1,v_2) \mid 
                        \mtenv ; \mtenv \vdash v_1 : \tau \wedge 
                        \mtenv ; \mtenv \vdash v_2 : \tau \wedge \dots \}
\]
For succinctness, we leave this implicit when we define the rest of the value relation.
As a first try, we naively construct the logical relation in the same way we constructed the logical predicates in the previous sections:
\[
\arraycolsep=0pt
\begin{array}{rl}
  \vpred{bool}                 ={} & \{ (\true,\true), (\false,\false) \} \\
  \vpred{\tarrow{\tau}{\tau'}} ={} & \left\{
                                 \begin{array}{>{$}p{9.5cm}<{$}}
                                   (\tlabs{x}{\tau}{e_1},\tlabs{x}{\tau}{e_2}) \mid \\
                                   \multicolumn{1}{r}{\forall (v_1,v_2) \in \vpred{\tau}\ldotp (\subst{e_1}{v_1}{x},\subst{e_2}{v_2}{x}) \in \epred{\tau'}}
                                 \end{array}
\right\}
\end{array}
\]
The value interpretation of the function type is defined based on the slogan for logical relations: ``Related inputs to related outputs.''
If we had chosen to use equal inputs rather than related, then our definition would be more restrictive than necessary.

We did not define a value interpretation for type variables in the previous sections, so let us try to push on without defining that part.

The next type is $\forall \alpha\ldotp \tau$.
When we define the value interpretation, we consider the elimination forms which in this case is type application (cf.\ the guidelines in Section~\ref{subsec:categories-lr}).
Before we proceed, let us revisit Theorem~\ref{thm:free-thm-three}, the last free theorem of Section~\ref{subsec:free-thm}.
There are some important points to note in this free theorem.
First of all, we want to be able to apply $\Lambda$-terms to different types, so we will have to pick two different types in our value interpretation.
Further in the lambda-abstraction case, we pick related values, so perhaps we should pick related types here.
However, we do not have a notion of related types, and the theorem does not relate the two types, so relating them might not be a good idea after all.
With these points in mind, we make a first attempt at constructing the value interpretation of $\forall \alpha\ldotp \tau$:
\[
  \vpred{\forall \alpha\ldotp \tau} = \{(\tLabs{e_1}, \tLabs{e_2}) \mid \forall \tau_1,\tau_2\ldotp (\subst{e_1}{\tau_1}{\alpha},\subst{e_2}{\tau_2}{\alpha}) \in \epred{\subst{\tau}{?}{\alpha}} \}
\]
Now the question is what type to relate the two expressions under.
We need to substitute $?$ for some type, but if we use either $\tau_1$ or $\tau_2$, then the well-typedness requirement will be broken for one of the two expressions.
The solution is for each type variable to keep track of the two types we would have liked to substitute, i.e.\ $\tau_1$ and $\tau_2$.
This means that we leave $\alpha$ free in $\tau$, but we can close it off with the types we keep track of to fulfil the well-typedness requirement of the logical relation.
To this end, we introduce a \emph{relational substitution} which keeps track of related types for each free type variable:
\[
  \rho = \{ \alpha_1 \mapsto (\tau_{11},\tau_{12} ), \dots \}
\]
We make the relational substitution available by passing it around in the interpretations. For instance, for the type abstraction we have:
\[
  \vprep{\forall \alpha\ldotp \tau} = \{(\tLabs{e_1},\tLabs{e_2}) \mid \forall \tau_1 , \tau_2\ldotp (\subst{e_1}{\tau_1}{\alpha},\subst{e_2}{\tau_2}{\alpha}) \in \eprep[\extsub{\rho}{\alpha}{(\tau_1,\tau_2)}]{\tau} \}
\]
Notice how the definition now is parameterized with $\rho$, but also how $\rho$ is extended with $\alpha$ mapping to $\tau_1$ and $\tau_2$ before it is passed to the expression interpretation.
This means that $\alpha$ is free in $\tau$ which means that we need to interpret type variables in the value interpretation:
\[
\vprep{\alpha} = \{(v_1,v_2) \mid \rho(\alpha) = (\tau_1,\tau_2) \dots\}
\]
However, the big question is: How should these values be related?
To find the answer, we look to Theorem~\ref{thm:free-thm-three} again.
In the theorem, the expression $e$ is contextually related to itself when applied to different types and values.
That is $e[\tau_1]\; v_1$ is contextually equivalent to $e[\tau_2]\; v_2$ which means that $e$ must treat $v_1$ and $v_2$ the same.
This would suggest that, $v_1$ and $v_2$ should somehow be related under type $\alpha$.
To achieve this in the value interpretation, we pick a relation $R$ in the interpretation of type abstractions.
Relation $R$ relates values of type $\tau_1$ and $\tau_2$ which in the example of the free theorem would allow us to relate $v_1$ and $v_2$.
An example of a relation would be $R_{\var{ex}}=\{(1,\true)\}$).
We need such a relation when we interpret a type variable $\alpha$, so we make $R$ part of the relational substitution.
For instance, if we picked $R_{\var{ex}}$ for $\alpha$ and made it part of the relational substitution, then it would be of the form $\extsub{\rho}{\alpha}{(\int,\bool,R_{\var{ex}})}$ where $\rho$ is the previous relational substitution.
In the interpretation of type variables, we simply look-up $R$ in the relational substitution and require the values to be related in the relation.
With this in mind, we update the interpretation of type abstractions:
\[
  \arraycolsep=0pt
  \begin{array}{rl}
    \vprep{\forall \alpha\ldotp \tau} ={} & 
                                            \left\{ (\tLabs{e_1},\tLabs{e_2}) \middle|
                                            \begin{array}{>{$}p{8cm}<{$}}
                                            \!\begin{multlined}
                                              \forall \tau_1 , \tau_2, R \in \Rel[\tau_1,\tau_2]n\ldotp \\
                                              (\subst{e_1}{\tau_1}{\alpha},\subst{e_2}{\tau_2}{\alpha}) \in \eprep[\extsub{\rho}{\alpha}{(\tau_1,\tau_2,R)}]{\tau}
                                          \end{multlined}
                                            \end{array}
\right\}
  \end{array}
\]
The relational substitution is now also extended with $R$, and we require $R$ to be in $\Rel[\tau_1,\tau_2]$.
That is $\tau_1$ and $\tau_2$ must be closed values and $R$ only relates well-typed closed values.
\[
  \Rel[\tau_1,\tau_2] = \left\{R \in \curly{P}(Val \times Val) \middle|
    \begin{array}{l}
      \mtenv \vdash \tau_1 \wedge \mtenv \vdash \tau_2 \wedge\\
      \forall (v_1,v_2) \in R\ldotp \mtenv \vdash v_1 : \tau_1 \wedge \mtenv \vdash v_2 : \tau_2
    \end{array}
\right\}
\]
For $R_{\var{ex}}$, we have $R_{\var{ex}} \in \Rel[\int,\bool]$ as $\int$ and $\bool$ are closed types and $1$ and $\true$ are closed values of type $\int$ and $\true$, respectively.
In the interpretation of a type variable $\alpha$, we require the pair of values to be related in the relation for that type variable.
That is, they must be related under the  relation we picked when interpreting $\forall \alpha\ldotp \tau$:
\[
  \vprep{\alpha} = \{ (v_1,v_2) \mid \rho(\alpha) = (\tau_1,\tau_2,R) \wedge (v_1,v_2) \in R \}
\]
For convenience, we introduce the following notation: Given
\[
  \rho = \{\alpha_1 \mapsto (\tau_{11},\tau_{12},R_1), \alpha_2 \mapsto (\tau_{21},\tau_{22},R_2),\; \dots\; \}
\]
define the following projections:
\begin{align*}
  \rho_1 & = \{\alpha_1 \mapsto \tau_{11}, \; \alpha_2 \mapsto \tau_{21},\; \dots \;\} \\
  \rho_2 & = \{\alpha_1 \mapsto \tau_{12}, \; \alpha_2 \mapsto \tau_{22},\; \dots \;\} \\
  \rho_R & = \{\alpha_1 \mapsto R_1, \; \alpha_2 \mapsto R_2,\; \dots \;\} 
\end{align*}
Notice that $\rho_1$ and $\rho_2$ now are type substitutions, so we write $\rho_1(\tau)$ to mean $\tau$ where all the type variables mentioned in the substitution have been substituted with the appropriate types.
We can now write the value interpretation for type variables in a more succinct way:
\[
  \vprep{\alpha} = \rho_R(\alpha)
\]
%TODO: include something like: ``When we use our logical relation to prove results such as the free theorem above, then we get to choose the relation $R$ which means that we get to choose what values are related at $\alpha$!''? and ``The power of parametricity is that you pick types for $\alpha$ as well as the notion values of those types are related under.
%TODO: further: ``If we want to prove something about a term with type $\forall \alpha. \; \tau$, then we get to pick the the relation R. Whereas if we want to prove a term is of a polymorphic type, then we have the relation as a proof obligation and have to show it for an arbitrary R.
We need to add $\rho$ to the other parts of the value interpretation as well. Moreover, as we now interpret open types, we require the pairs of values in the relation to be well typed under the type closed off using the relational substitution. So all value interpretations have the form
\[
  \vprep{\tau} = \{(v_1,v_2) \mid \mtenv ; \mtenv \vdash v_1 : \rho_1(\tau) \wedge \mtenv ; \mtenv \vdash v_2 : \rho_2(\tau) \wedge\dots \}
\]
We further need to close of the type annotation of the variable in functions, so our value interpretations end up as:

\[
  \arraycolsep=0pt
  \begin{array}{rl}
    \vprep{bool}                 ={} & \{ (\true,\true), (\false,\false) \} \\
    \vprep{\tarrow{\tau}{\tau'}} ={} & \left\{
                                   \begin{array}{>{$}p{9.5cm}<{$}}(\tlabs{x}{\rho_1(\tau)}{e_1},\tlabs{x}{\rho_2(\tau)}{e_2}) \mid \\
                                      \multicolumn{1}{r}{\forall (v_1,v_2) \in \vprep{\tau}\ldotp (\subst{e_1}{v_1}{x},\subst{e_2}{v_2}{x}) \in \eprep{\tau'}} 
                                   \end{array}\right\}
  \end{array}
\]
With the above challenge resolved, the remainder of the definitions are fairly straightforward.
The expression interpretation is defined as follows:
\[
\begin{array}{l}
  \eprep{\tau} ={}  \left\{
(e_1,e_2) \middle|
\begin{array}{l}
    \mtenv ; \mtenv \vdash e_1 : \rho_1(\tau) \wedge 
    \mtenv ; \mtenv \vdash e_2 : \rho_2(\tau) \wedge \\
    \exists v_1,v_2 \ldotp e_1 \evaltos v_1 \wedge
    e_2 \evaltos v_2 \wedge
    (v_1,v_2) \in \vprep{\tau} 
    \end{array}
\right\}
\end{array}
\]
In System F, the environment is extended with a type context $\Gamma$, so we need an interpretation for it.
\begin{align*}
  \dpred{\mtenv}         &= \{\emptyset \} \\
  \dpred{\Delta, \alpha} &= \{\extsub{\rho}{\alpha}{(\tau_1,\tau_2,R)} \mid 
                                 \rho \in \dpred{\Delta} \wedge 
                                 R \in \Rel[\tau_1,\tau_2] \} \\
\end{align*}
The interpretation of the type context is simply a relational substitution for the variables in the context.
We also need an interpretation for the typing context $\Gamma$ just like we had one in the previous sections.
\begin{align*}
  \gprep{\mtenv}         &= \{\emptyset \} \\
  \gprep{\Gamma, x:\tau} &= \{\extsub{\gamma}{x}{(v_1,v_2)} \mid 
                                 \gamma \in \gprep{\Gamma} \wedge
                                 (v_1,v_2) \in \vprep{\tau}\}
\end{align*}
We need the relational substitution in the interpretation of $\Gamma$, because $\tau$ might contain free type variables now.
%In \gpred{\Gamma, x:\tau} \tau may have free type variable in it which is why we need \rho. recall \Delta ; \Gamma \vdash e : \tau (?)
We introduce a convenient notation for the projections of $\gamma$ similar to the one we did for $\rho$: %TODO: Write something about the fact that we now want to close off pairs of expressions, so we get a substitution on pairs.
\[
  \gamma = \{x_1 \mapsto (v_{11},v_{12}), x_2 \mapsto (v_{21},v_{22}),\; \dots\; \}
\]
Define the projections as follows:
\begin{align*}
  \gamma_1 & = \{x_1 \mapsto v_{11}, \; x_2 \mapsto v_{21},\; \dots \;\} \\
  \gamma_2 & = \{x_1 \mapsto v_{12}, \; x_2 \mapsto v_{22},\; \dots \;\}
\end{align*}
We are now ready to define when two terms are logically related.
We define it in a similar way to the logical predicate we already have defined.
First we pick $\rho$ and $\gamma$ to close off the expressions, then we require the closed off expressions to be related under the expression interpretation of the type in question\renewcommand{\lreq}[2]{\equivalence{#1}{}{#2}}\footnote{From now on we drop the superscript for logically related terms and just write $\lreq{}{}$.}:
\[
  \Delta ; \Gamma \vdash \lreq{e_1}{e_2} : \tau \eqdef{} \left\{
    \begin{array}{l}
      \Delta ; \Gamma \vdash e_1 : \tau \wedge\\
      \Delta ; \Gamma \vdash e_2 : \tau \wedge \\
      \forall \rho \in \dpred{\Delta}, \gamma \in \gprep{\Gamma} \ldotp\\
      \quad(\rho_1(\gamma_1(e_1)),\rho_2(\gamma_2(e_2))) \in \eprep{\tau}
    \end{array}
\right.
\]
With our logical relation defined, the first thing to prove is the fundamental property:
\begin{theorem}[Fundamental Property]
  \label{thm:system-f-ftlr}
  If $\Delta; \Gamma \vdash e : \tau$ then $\Delta; \Gamma \vdash \lreq{e}{e} : \tau$
\end{theorem}
This theorem may seem a bit mundane, but it is actually quite strong.
In the definition of the logical relation, $\Delta$ and $\Gamma$ can be seen as maps from place holders that need to be replaced in the expression.
When we choose $\rho$ and $\gamma$, we can pick different types and terms to put in the expression which can give us two very different programs.

In some presentations, this is also known as the parametricity lemma.
It may even be stated without the short-hand notation for equivalence we use here.

We use a number of \emph{compatibility lemmas} to prove the theorem.
Each of the compatibility lemmas correspond to a typing rule and essentially prove one case of the induction proof for the theorem.
%''Parametricity'' \Delta; \Gamma <- Place holders for input. When we close off the terms we get two different programs.
%Prove directly via induction on typing judgment. One case per typing rule.
\subsection{Compatibility Lemmas}
We state a compatibility lemma for each of the typing rules.
Each of the lemmas correspond to a case in the induction proof of the Fundamental Property (Theorem~\ref{thm:system-f-ftlr}), so the theorem will follow directly from the compatibility lemmas.
We state the compatibility lemmas as rules to illustrate the connection to the typing rules:
\begin{enumerate}
\item $\inferrule*{ }
                 {\Gamma ; \Delta \vdash \lreq{\true}{\true} : bool}$ % \approx \true ``added to typing rule'' (?)
\item $\inferrule*{ }{\Gamma ; \Delta \vdash \lreq{\false}{\false} : bool}$
\item $\inferrule*{ }{\Gamma ; \Delta \vdash \lreq{x}{x} : \Gamma(x)}$
\item $\inferrule*{\Delta;\Gamma \vdash \lreq{e_1}{e_2} : \tarrow{\tau'}{\tau} \and
                   \Delta;\Gamma \vdash \lreq{e_1'}{e_2'} : \tau'}
                  {\Delta;\Gamma \vdash \lreq{e_1 \; e_1'}{e_2 \; e_2'} : \tau}$
%slightly more general than we need, different things on both sides, gives us FP, but we can use them to prove more.
\item $\inferrule*{\Delta; \Gamma, x:\tau \vdash \lreq{e_1}{e_2} : \tau'}
                  {\Delta; \Gamma \vdash \lreq{\tlabs{x}{\tau}{e_1}}{\tlabs{x}{\tau}{e_2}} : \tarrow{\tau}{\tau'}}$
\item $\inferrule*{\Delta; \Gamma \vdash \lreq{e_1}{e_2} : \forall \alpha\ldotp \tau \and
                   \Delta \vdash \tau'}
                  {\Delta; \Gamma \vdash \lreq{e_1[\tau']}{e_2[\tau']} : \subst{\tau}{\tau'}{\alpha}}$
%TODO: Add if? (then delete the text that says it has been omitted
\end{enumerate}
The rule for \texttt{if} has been omitted here.
Further notice, some of the lemmas are more general than what we actually need.
Take for instance the compatibility lemma for expression application.
To prove the fundamental property, we really just needed to have the same expressions on both sides of the equivalence.
However, the generalized version can also be used to prove that the logical relation is sound with respect to contextual equivalence.

We prove the compatibility lemma for type application here and leave the remaining proofs to the reader.
To prove the compatibility lemma for type application, we are going to need the following lemma
\begin{lemma}[Compositionality]
  \label{lem:system-f-compositionality}
  If
\begin{itemize}\setlength\itemsep{0em} \renewcommand\labelitemi{--}
\item $\Delta \vdash \tau'$
\item $\Delta, \alpha \vdash \tau$
\item $\rho \in \dpred{\Delta}$
\item $R=\vprep{\tau'}$
\end{itemize}
 then
\[
  \vprep{\subst{\tau}{\tau'}{\alpha}} = \vprep[\extsub{\rho}{\alpha}{(\rho_1(\tau'),\rho_2(\tau'),R)}]{\tau} \qedhere
\]
\end{lemma}
The lemma says: syntactically substituting some type for $\alpha$ in $\tau$ and then interpreting it is the same as semantically substituting the type for $\alpha$.
To prove this lemma, we would need to show $\vprep{\tau} \in \Rel[\rho_1(\tau),\rho_2(\tau)]$ which is fairly easy given how we have defined our value interpretation.
% The proof is by induction over the structure of $\tau$ and it needs to keep \rho general for the type abstraction case.
\begin{proof}[Proof. (Compatibility, Lemma 6)]
What we want to show is 
\[
  \inferrule*{\Delta; \Gamma \vdash \lreq{e_1}{e_2} : \forall \alpha\ldotp \tau \and
              \Delta \vdash \tau'}
             {\Delta; \Gamma \vdash \lreq{e_1[\tau']}{e_2[\tau']} : \subst{\tau}{\tau'}{\alpha}}
\]
So we assume 
\begin{enumerate}[label=(\arabic*)]
\item \label{item:comp-ass-one} $\Delta; \Gamma \vdash \lreq{e_1}{e_2} : \forall \alpha\ldotp \tau$
\item \label{item:comp-ass-two} $\Delta \vdash \tau'$
\end{enumerate}
According to our definition of the logical relation, we need to show three things:

\begin{enumerate}[label=\roman*.]
\item $\Delta ; \Gamma \vdash e_1[\tau'] : \subst{\tau}{\tau'}{\alpha}$
\item $\Delta ; \Gamma \vdash e_2[\tau'] : \subst{\tau}{\tau'}{\alpha}$ 
\item \label{item:comp-pb} $\forall \rho \in \dpred{\Delta}\ldotp \forall \gamma \in \gprep{\Gamma}\ldotp (\rho_1(\gamma_1(e_1[\tau'])),\rho_2(\gamma_2(e_2[\tau']))) \in \eprep{\subst{\tau}{\tau'}{\alpha}}$
   \end{enumerate}
   The two first follows from the well-typedness part of \ref{item:comp-ass-one} together with \ref{item:comp-ass-two} and the appropriate typing rule.
   So it only remains to show \ref{item:comp-pb}

Suppose we have $\rho \in \dpred{\Delta}$ and $\gamma \in \gprep{\Gamma}$. We then need to show:
\[
  (\rho_1(\gamma_1(e_1[\tau'])),\rho_2(\gamma_2(e_2[\tau']))) \in \eprep{\subst{\tau}{\tau'}{\alpha}}
\]
By the definition of the $\curly{E}$-relation, we need to show that the two terms step to two related values.
We keep this goal in mind and turn our attention to \ref{item:comp-ass-one}. By definition it gives us:
\[
\forall \rho \in \dpred{\Delta}\ldotp \forall \gamma \in \gprep{\Gamma}\ldotp (\rho_1(\gamma_1(e_1)),\rho_2(\gamma_2(e_2))) \in \eprep{\forall \alpha \ldotp \tau}
\]
If we instantiate this with our $\rho$ and $\gamma$, then we get 
\[
  (\rho_1(\gamma_1(e_1)),\rho_2(\gamma_2(e_2))) \in \eprep{\forall \alpha \ldotp \tau}
\]
which means that $e_1$ and $e_2$ evaluate to some value $v_1$ and $v_2$ where $(v_1,v_2) \in \vprep{\forall \alpha\ldotp \tau}$.
This means that $v_1$ and $v_2$ must be of type of type $\forall \alpha\ldotp \tau$.
From this, we know that $v_1$ and $v_2$ are type abstractions, so there must exist $e_1'$ and $e_2'$ such that $v_1 = \tLabs{e_1'}$ and $v_2 = \tLabs{e_2'}$.
We can now instantiate $(v_1,v_2) \in \vprep{\forall \alpha\ldotp \tau}$ with two types and a relation.
We choose $\rho_1(\tau')$ and $\rho_2(\tau')$ as the two types for the instantiation and $\vprep{\tau'}$ as the relation\footnote{Here we use $\vprep{\tau} \in \Rel[\rho_1(\tau),\rho_2(\tau)]$ to justify using the value interpretation as our relation.}.
This gives us
\[
  (\subst{e_1'}{\rho_1(\tau')}{\alpha},\subst{e_2'}{\rho_2(\tau')}{\alpha}) \in \eprep[\extsub{\rho}{\alpha}{(\rho_1(\tau'),\rho_2(\tau'),\vprep{\tau'})}]{\tau}
\]
For convenience, we write $\rho' = \extsub{\rho}{\alpha}{(\rho_1(\tau'),\rho_2(\tau'),\vprep{\tau'})}$.
By definition of the $\curly{E}$-relation, we know that $\subst{e_1'}{\rho_1(\tau)}{\alpha}$ and $\subst{e_2'}{\rho_2(\tau)}{\alpha}$ evaluate to some values $v_{1_f}$ and $v_{2_f}$, respectively, where $(v_{1_f},v_{2_f}) \in \vprep[\rho']{\tau}$.

Let us take a step back and see what we have done. We have argued that the following evaluation takes place
\begin{align*}
  \rho_i (\gamma_i(e_i)) [\rho_i(\tau')] & \evaltos (\tLabs{e_i'})[\rho_i(\tau')] \\
                      & \evalto \subst{e_i'}{\rho_i(\tau')}{\alpha} \\
                      & \evaltos v_{i_f}
\end{align*}
for $i=1,2$.
The single step in the middle is justified by the type application reduction.
The remaining steps are justified in our proof above.
If we further note that $\rho_i (\gamma_i (e_i [\tau'])) \equiv \rho_i (\gamma_i (e_i))[\rho_i(\tau')]$, then we have shown that the two expressions from our goal in fact do evaluate to two values, and they are related.
More precisely we have:
\[
  (v_{1_f},v_{2_f}) \in \vprep[\rho']{\tau}
\]
but that is not exactly what we wanted them to be related under.
We are, however, in luck and can apply the compositionality lemma (Lemma~\ref{lem:system-f-compositionality}) to obtain
\[
  (v_{1_f},v_{2_f}) \in \vprep{\subst{\tau}{\tau'}{\alpha}}
\]
which means that they are related under the relation we needed.
\end{proof}

We call theorems that follows as a consequence of parametricity for free theorems.
The theorems from Section~\ref{subsec:free-thm} are examples of free theorems.
Next we will prove Theorem~\ref{thm:free-thm-identity}: All expressions of type $\forall \alpha\ldotp \tarrow{\alpha}{\alpha}$ must be the identity function.
The theorem even says that any function of this type will terminate. This is, however, trivial as System-F is a terminating language\footnote{For more on this see \emph{Types and Programming Languages} by Benjamin Pierce\citep{Pierce:types-and-pl}.}.
In a non-terminating language such as System F with recursive types, we would state a weaker theorem where the expression would only have to be the identity function when it terminates.
That is, divergence would also be acceptable behaviour.
\begin{proof}[Proof of Theorem~\ref{thm:free-thm-identity}]
\newcommand{\aaa}{\ensuremath{\forall \alpha\ldotp \tarrow{\alpha}{\alpha}}}
\newcommand{\mt}{\ensuremath{\emptyset}}
\newcommand{\env}{\ensuremath{\extsub{ }{\alpha}{(\tau,\tau,R)}}}
\newcommand{\taa}{\ensuremath{\tarrow{\alpha}{\alpha}}}
From the fundamental property and the well-typedness of $e$, we know $\mtenv \vdash \lreq{e}{e} : \aaa$.
By definition this gives us
\[
\forall \rho \in \dpred{\Delta}\ldotp \forall \gamma \in \gprep{\Gamma}\ldotp (\rho_1(\gamma_1(e)),\rho_2(\gamma_2(e))) \in \eprep{\aaa}
\]
We instantiate this with an empty $\rho$ and an empty $\gamma$ to get $(e,e) \in \eprep[\mt]{\aaa}$.
By the definition of $\curly{E}$, we know that $e$ evaluates to some value $F$ and $(F,F) \in \vprep[\mt]{\aaa}$.
As $F$ is a value of type $\aaa{}$, we know $F=\tLabs{e_1}$ for some $e_1$.
Now instantiate $(F,F) \in \vprep[\mt]{\aaa}$ with the type $\tau$ twice and the relation $R=\{(v,v)\}$ (Note: $R \in \Rel[\tau,\tau]$).
This gives us $(\subst{e_1}{\tau}{\alpha},\subst{e_1}{\tau}{\alpha})\in \eprep[\env]{\taa}$.

Let us take a quick intermission from the proof.
To a large extend, the proof of a free theorem is just unfolding of definitions.
However in order for the definitions to line up, it is important that the relation $R$ is picked correctly which makes this an important non-trivial part of the proof.
Before we chose the relation in this proof, we picked two types based on the theorem we want to show.
In the theorem, we instantiate $e$ with $\tau$, so we picked $\tau$ for the proof.
Likewise with $R$: in the theorem $v$ is the argument for the function with the domain $\alpha$, so we picked the singleton relation $\{(v,v)\}$.
Picking the correct relation requires some work, but the statement can guide the decision.
Finally, if you are not sure what to pick for $R$ try something and see if it works out.
If it does not work, then you can simply pick a new relation and you may even have learned something to guide your next pick.
Intermission over.

From $(\subst{e_1}{\tau}{\alpha},\subst{e_1}{\tau}{\alpha})\in \eprep[\env]{\taa}$, we know that $\subst{e_1}{\tau}{\alpha}$ evaluates to some value $g$ and $(g,g)\in \vprep[\env]{\taa}$.
From the type of $g$, we know that it must be a $\lambda$-abstraction, so $g=\tlabs{x}{\tau}{e_2}$ for some expression $e_2$.
Now instantiate $(g,g)\in \vprep[\env]{\taa}$ with $(v,v) \in \vprep[\env]{\alpha}$ to get $(\subst{e_2}{v}{x},\subst{e_2}{v}{x}) \in \eprep[\env]{\alpha}$.
From this we know that $\subst{e_2}{v}{x}$ steps to some value $v_f$ and $(v_f,v_f) \in \vprep[\env]{\alpha}$.
We have that $\vprep[\env]{\alpha} \equiv R$ so $(v_f,v_f) \in R$ which mean that $v_f = v$ as $(v,v)$ is the only pair in $R$.

Now let us take a step back and consider what we have shown above.
\begin{align*}
  e[\tau] \; v & \evaltos F [\tau] \; v \\
               & \equiv (\tLabs{e_1}) [\tau] \; v \\
               & \evalto (\subst{e_1}{\tau}{\alpha}) \; v \\
               & \evaltos g \; v\\
               & \equiv (\tlabs{x}{\tau}{e_2}) \; v \\
               & \evalto \subst{e_2}{v}{x} \\
               & \evaltos v_f \\
               & \equiv v
\end{align*}
First we argued that $e[\tau]$ steps to some $F$ and that $F$ was a type abstraction, $\tLabs{e_1}$.
Then we performed the type application to get $\subst{e_1}{\tau}{\alpha}$.
We then argued that this steps to some $g$ of the form $\tlabs{x}{\tau}{e_2}$ which further allowed us to do a $\beta$-reduction to obtain $\subst{e_2}{v}{x}$.
We then argued that this reduced to $v_f$ which was the same as $v$.
In summation we argued $e[\tau] \; v \evaltos v$ which is the result we wanted.
\end{proof}
%We can take a number of steps and end up with v.
%choice of R was critical.
%\forall \alpha. \alpha easier to show (?)
\subsection{Exercises}
\begin{enumerate}
\item Prove the following free theorem:
  \begin{theorem}
    If $\mtenv; \mtenv \vdash e : \forall \alpha\ldotp (\tarrow{(\tarrow{\tau}{\alpha})}{\alpha})$ and 
       $\mtenv; \mtenv \vdash k : \tarrow{\tau}{\tau_k}$, then
    \[
      \mtenv; \mtenv \vdash \lreq{e[\tau_k] \; k}{k(e[\tau] \; \tlabs{x}{\tau}{x})} : \tau_k
      \qedhere
%''\mtenv; \mtenv \vdash'' not in my notes
    \]
  \end{theorem}
  This theorem is a simplified version of one found in \emph{Theorems For Free} by Philip Wadler\citep{Wadler:free-theorems}. 
%TODO: insert proper reference (http://ttic.uchicago.edu/~dreyer/course/papers/wadler.pdf)
\end{enumerate}
\section{Existential types}
In this section, we extend \STLC{} with existential types and show how to make a logical relation for that language.

An existential type is reminiscent of a Java interface.
It describes some functionality but leaves out the implementation.
An implementation of an existential type must provide the functions and constants specified in the type.
This means that there can be many different implementations of the same existential type.
When an existential type is used, the user interacts with the exposed constants and functions without knowing the actual implementation.
% Use this technique to implement abstract data types.

Take for example a stack.
We would expect a stack to have the following functions:
\begin{description}[font=\ttfamily]
  \item[mk] creates a new stack.
  \item[push] puts an element on the top of the stack.
    It takes a stack and an element and returns the resulting stack.
  \item[pop] removes the top element of the stack.
    It takes a stack and returns the new stack along with the element that was popped from it.
\end{description}
An interface would define the above signature which you then would go off and implement\footnote{There is a famous paper called \emph{Abstract Data Types Have Existential Type} from '88 by Mitchell and Plotkin~\citep{Mitchell:abstract-types}.
  The title says it all.}.
If we wanted to write an interface for a stack, it would look like this (this is meant to be suggestive, so it is in a non-formal notation):
\begin{align*}
  \Istack = \exists \alpha\ldotp \langle & \mk: \tarrow{1}{\alpha}, \\
                                 & \push: \tarrow{\alpha \times \int}{\alpha}, \\
                                 & \pop: \tarrow{\alpha}{\alpha \times \int} \rangle
\end{align*}
where $\alpha$ is the type of the stack used in the actual implementation.
This means that $\mk$ is a function that takes unit and returns a stack.
The $\push$ function takes a stack and an element and returns a stack.
Finally, the $\pop$ function takes a stack and returns a stack and an element.
The $\Istack$ type all the $\alpha$'s which means that a client cannot see the actual type of the stack, so they cannot know how it is actually implemented.

We formally write existential types in a similar fashion to how we wrote universal types:
\[
  \exists \alpha\ldotp \tau
\]
Here $\tau$ is the same as the record in the stack example.
The interface is just a type, so now we need to define how one implements something of an existential type.
If we were to implement the stack interface, then we would implement a package of functions that are supposed to be used together.
This could look something like (again this is meant to be suggestive):
%implement stack creating package:
\begin{align*}
  \epack\; &\tarray[\int], \\
         & \langle \lambda x:\_\ldotp\; \text{...}\;,\\
         & \; \lambda x:\_\ldotp\; \text{...}\;,\\
         & \; \lambda x:\_\ldotp\; \text{...}\; \rangle
\end{align*}
Here $\tarray[\int]$ is the type we want to use for the concrete implementation and the record of functions is the concrete implementation that uses $\tarray[\int]$ to implement a stack.

Existential types hide the implementation from the client, but that does not mean that they are the same.
For instance, a $\Istack$ implementation could use $\int$ for $\alpha$.
In this case $\mk$ could return $0$, $\push$ would throw away the $\int$ argument, and $\pop$ would just return a pair of its $\alpha$ argument, which is an $\int$ argument here.
This implementation clearly doesn't correspond to what we consider a stack, but it has the correct type.
We would expect that all implementations of the interface that correspond to a stack would observationally behave the same.
That is, we would expect them to be contextually equivalent.
We want to use a logical relation to be able to establish contextual equivalence.

\subsection{\STLC{} with Existential Types}
The syntactic additions to the language are \texttt{pack} and \texttt{unpack}:
We briefly introduce the additions to the syntax and semantics:
\begin{align*}
  \tau & ::= \ldots \mid \exists \alpha\ldotp \tau \\
  e & ::= \ldots \mid \pack{\tau}{e}{\exists \alpha\ldotp \tau} \mid \unpack{\alpha}{x}{e}{e} \\
  v & ::= \ldots \mid \pack{\tau}{v}{\exists \alpha\ldotp \tau} \\
  E & ::= \ldots \mid \pack{\tau}{E}{\exists \alpha\ldotp \tau} \mid \unpack{\alpha}{x}{E}{e}
\end{align*}
New packages are constructed with \texttt{pack} which makes it the introductory form for the existential type.
A package consists of a \emph{witness type} and an implementation of the existential type.
Packages are opened by \texttt{unpack} which allows a client to use the components in the package.
This makes \texttt{unpack} the elimination form.
The interaction between the two new constructs is expressed in the following evaluation rule:
\begin{mathpar}
  \inferrule{ }{\unpack{\alpha}{x}{\pack{\tau'}{v}{\exists \alpha\ldotp \tau}}{e} \evalto \subst{\subst{e}{\tau'}{\alpha}}{v}{x}}
\end{mathpar}
The form of the typing judgement remains the same but we add the following two rules:
\begin{mathpar}
  \inferrule*{\Delta; \Gamma \vdash e : \subst{\tau}{\tau'}{\alpha} \and
    \Delta \vdash \tau'}
             {\Delta; \Gamma \vdash \pack{\tau'}{e}{\exists \alpha\ldotp\tau} : \exists \alpha\ldotp\tau}
\and
\inferrule*{\Delta; \Gamma \vdash e_1 : \exists \alpha\ldotp \tau \and
               \Delta,\alpha;\Gamma,x:\tau \vdash e_2 : \tau_2 \and
               \Delta \vdash \tau_2}
             {\Delta; \Gamma \vdash \unpack{\alpha}{x}{e_1}{e_2} : \tau_2}
\end{mathpar}
Intuitively, the typing rule of \texttt{pack} says that provided an implementation of the existential type that implementation has to be well-typed when the witness type is plugged in for $\alpha$.
%We write '$as \exists \alpha\ldotp \tau$' to know what interface has been implemented.
It tells what type we substitute into.
In the typing rule for \texttt{unpack}, it is important that $\alpha$ is not free in $\tau_2$ which is ensured by $\Delta \vdash \tau_2$.
This makes sure that the package actually hides the witness type.
The witness type can be pulled out of the package within a certain scope using \texttt{unpack}.
If $\alpha$ could be returned from that scope, then it would be exposed to the outer world which would defeat the purpose of hiding it.
The \texttt{unpack}-expression takes out the components of $e_1$ and calls them $\alpha$ and $x$.
The two components can then be used in the body, $e_2$, of the \texttt{unpack}-expression, but they are used with the aliases given to them by \texttt{unpack}.

\subsection{Example}
\label{subsec:exi-example}
For the remainder of this section, we will use the following example:
Take this existential type
\[
  \tau = \exists \alpha\ldotp \alpha \times (\tarrow{\alpha}{\bool})
\]
and the following two expressions that we for now claim have type $\tau$
\begin{align*}
  e_1 & = \pack{\int}{\tuple{1, \tlabs{x}{\int}{x=0}}}{\tau} \\
  e_2 & = \pack{\bool}{\tuple{\true, \tlabs{x}{\bool}{\enot\; x}}}{\tau}
\end{align*}
where $\int$ and $\bool$ the \emph{witness types}, respectively.
We further claim that these two implementations of $\tau$ are equivalent, and our goal in this note is to show this.

We start by verifying that the two expressions are indeed well-typed and has type $\tau$ by constructing their type derivation trees:
\[
  \inferrule*{\inferrule*{\inferrule*{ }
                                     {\mtenv; \mtenv \vdash 1 : \int}
                          \and
                          \inferrule*{\inferrule*{\inferrule*{ }
                                                             {\mtenv; x: \int \vdash x : \int}
                                                  \and
                                                  \inferrule*{ }
                                                             {\mtenv; x: \int \vdash 0 : \int}}
                                                 {\mtenv; x: \int \vdash x = 0 : \bool}}
                                     {\mtenv;\mtenv \vdash \tlabs{x}{\int}{x=0} : \tarrow{\int}{\bool}}}
                         {\mtenv;\mtenv \vdash \tuple{1,\tlabs{x}{\int}{x=0}} : \int \times (\tarrow{\int}{\bool})}
              \and
              \inferrule*{ }
                         {\mtenv \vdash \int}}
             {\mtenv ; \mtenv \vdash \pack{\int}{\tuple{1,\tlabs{x}{\int}{x=0}}}{\tau} : \tau}
\]
\[
  \inferrule*{\inferrule*{\inferrule*{ }
                                     {\mtenv; \mtenv \vdash \true : \bool}
                          \and
                          \inferrule*{\inferrule*{\inferrule*{ }
                                                             {\mtenv; x: \bool \vdash x : \bool}}
                                                 {\mtenv; x: \bool \vdash \enot\; x : \bool}}
                                     {\mtenv;\mtenv \vdash \tlabs{x}{\bool}{\enot \; x} : \tarrow{\bool}{\bool}}}
                         {\mtenv;\mtenv \vdash \tuple{\true, \tlabs{x}{\bool}{\enot\; x}} : \bool \times (\tarrow{\bool}{\bool})}
              \and
              \inferrule*{ }
                         {\mtenv \vdash \bool}}
             {\mtenv ; \mtenv \vdash \pack{\bool}{\tuple{\true, \tlabs{x}{\bool}{\enot\; x}}}{\tau} : \tau}
\]

To use a package, a client must open it with \texttt{unpack} first.
For instance, one could attempt the following unpack:
\[
  \texttt{unpack}\; \tuple{\alpha,p} = e_1 \;\texttt{in}\;  (\snd \; p)\;5
\]
This use of \texttt{unpack} does, however, expose the inner workings of $e_1$.
In particular, it exposes that $e_1$ uses integers as an internal representation.
The existential type was supposed to hide this information making the implementation opaque, but that is clearly not the case here.
The \texttt{unpack}-expression is even syntactically well-formed, but it is not well-typed.
The witness type $\int$ in $e_1$ is not exposed by \texttt{unpack} in the typing rule.
This means that the body of an \texttt{unpack} should use the witness type in an abstract way, i.e.\ it should use $\alpha$.
When we type the \texttt{unpack}-expression, $(\snd \; p)$ takes an argument of type $\alpha$; but in the above expression, we give it something of type $\int$ (if you don't see the issue, try to construct the type derivation tree).
This means that we only can give provide $(\fst \; p)$ as an argument, i.e.\ 
\[
\begin{array}{l}
  \eunpack \tuple{\alpha,p} = e_1 \ein \\
  \quad\color{red} \cancel{(\snd \; p) \; 5} \\
  \quad(\snd \; p)\;(\fst \; p)
\end{array}
\]
It should come as no surprise that $e_2$ only can be used in the exact same way:
\[
\begin{array}{l}
  \eunpack \tuple{\alpha,p} = e_2 \ein \\
  \quad(\snd \; p)\;(\fst \; p)
\end{array}
\]
It is also not type check if we return $(\fst \; p)$ (or $(\snd \; p)$ for that matter) as the body of an \texttt{unpack} cannot have the $\alpha$ free in it.

We can now informally argue why $e_1$ and $e_2$ are observationally equivalent.
In both expressions, there is only one value of type $\alpha$.
In $e_1$ it is $1$, and in $e_2$ it is $\true$.
This means that they are going to be the only related values, i.e.\ we will define $R=\{(1,\true)\}$
As we have informally argued, they are the only values the two functions can take as arguments.
This means that we can easily find the values that they can possibly return:
In $e_1$ it is $(\tlabs{x}{x=0})\; 1 \evalto \false$, and in $e_2$ it is $(\tlabs{x}{\enot x })\; true \evalto \false$.
The only value that is ever exposed from the package is $\false$.
If this claim is true, then it is impossible for a client to distinguish the two packages from each other.

\subsection{Logical Relation}
To formally argue that $e_1$ and $e_2$ from the previous section are contextually equivalent, we need a logical relation.
To this end, we extend our previous logical relation, so it also interprets $\exists \alpha\ldotp \tau$.
The values we relate are of the form $(\pack{\tau_1}{v_1}{\exists\alpha\ldotp\tau}, \; \pack{\tau_2}{v_2}{\exists\alpha\ldotp\tau})$ and as always our first instinct should be to look at the elimination form, so we want to consider $\unpack{\alpha}{x}{\pack{\tau_i}{v_i}{\exists\alpha\ldotp\tau}}{e_i}$ for $i=1,2$.
Now it would be tempting to relate the two bodies, but we get a similar issue to what we had for universal types.
If we relate the two bodies, then what type should we relate them under?
The type we get might be larger than the one we are interpreting which gives us a well-foundedness problem.
So by analogy we do not require that the two unpack expressions have related bodies.
Instead we relate $v_1$ and $v_2$ under some relation:
% Consider elim forms:
% unpack <\alpha,x> = pack \tau_1,v_1 as ... in e_1
% Size issue with e_1, the type of e_1 might be larger than what we got.
\[  \vprep{\exists\alpha\ldotp\tau} = \left\{
  \!\begin{multlined}
    \left(\pack{\rho_1(\tau_1)}{v_1}{\rho_1(\exists\alpha\ldotp\tau)},\pack{\rho_2(\tau_2)}{v_2}{\rho_2(\exists\alpha\ldotp\tau)}\right) \mid\\
  \exists R \in \Rel[\rho_1(\tau_1),\rho_2(\tau_2)]\ldotp (v_1,v_2) \in \vprep[\extsub{\rho}{\alpha}{(\rho_1(\tau_1),\rho_2(\tau_2),R)}]{\tau}
\end{multlined}
  \right\}
\]% Consider whether this should be related at e_1 and e_2. A: See the syntax. We require the right component in the pair to be a value.
%TODO: Rav nævnte, at der skulle være \rho_1 og \rho_2 om typerne i relational substitution, hvis et bestemt bevis skulle gå igennem. (dette er blevet ændret, tjek evt. med Lars, opgaven var en afleveringsopgave med existentials, hvor man skulle vise et compatibility lemma for pack)
The relation turns out to be somewhat dual to the one for universal types. Instead of saying $\forall \tau_1,\tau_2,R$, we say $\exists \tau_1,\tau_2,R$, but as we get $\tau_1$ and $\tau_2$ directly from the values, we omit them in the definition. We also relate the two values at $\tau$ and extend the relational substitution with the types we have for $\alpha$. Notice that we use $\rho$ to close of the type variables in the two values we related.

With this extension to the value interpretation, we are ready to show that $e_1$ and $e_2$ are logically related. We reuse the definition of logical equivalence we defined previously. What we wish to show is:

\begin{theorem}
\[
  \mtenv; \mtenv \vdash \lreq{e_1}{e_2} : \exists \alpha\ldotp \alpha \times (\tarrow{\alpha}{\bool})
  \qedhere
\]
\end{theorem}
\begin{proof}
  With an empty environment, this amounts to show $(e_1,e_2) \in \eprep[\emptyset]{\exists \alpha \ldotp \alpha \times (\tarrow{\alpha}{\bool})}$.
  To show this, we need to establish that $e_1$ and $e_2$ evaluate to some value and that these two values are related under the same type and relational substitution.
  The expressions $e_1$ and $e_2$ are \texttt{pack}-expressions, so they are already values.
  In other words, it suffices to show $(e_1,e_2) \in \vprep[\emptyset]{\exists \alpha \ldotp \alpha \times (\tarrow{\alpha}{\bool})}$.
  We now need to pick a relation and show that the implementations are related under $\alpha \times (\tarrow{\alpha}{\bool})$ that is
  \[
  (\tuple{1, \tlabs{x}{\int}{x=0}},\tuple{\true, \tlabs{x}{\bool}{\enot\; x}}) \in \vprep[\extsub{ }{\alpha}{(\int,\bool,R)}]{\alpha \times (\tarrow{\alpha}{\bool})}
  \]
  We pick $R=\{(1,\true)\}$ as the relation.
  To two tuples are related if their components are related\footnote{We defined this for logical predicates, but not for logical relations.}.
  In other words, we need to show the following:
  \begin{enumerate}
  \item \label{item:exi-item-one} $(1,\true) \in \vprep[\extsub{\emptyset}{ }{\alpha}{(\int,\bool,R)}]{\alpha}$
  \item \label{item:exi-item-two} $(\tlabs{x}{\int}{x=0},\tlabs{x}{\bool}{\enot\; x}) \in \vprep[\extsub{ }{\alpha}{(\int,\bool,R)}]{\tarrow{\alpha}{\bool}}$
  \end{enumerate}
  Item~\ref{item:exi-item-one} amounts to showing $(1,\true) \in R$ which is true.
  In order to show \ref{item:exi-item-two}, assume $(v_1,v_2) \in \vprep[\extsub{ }{\alpha}{(\int,\bool,R)}]{\alpha}$ which is the same as $(v_1,v_2) \in R$. 
  Due to our choice of $R$, we have $v_1 = 1$ and $v_2 = \true$.
  Now we need to show $(v_1 = 0,\enot\; v_2) \in \eprep[\extsub{ }{\alpha}{(\int,\bool,R)}]{\bool}$.
  Which means that we need to show that the two expressions evaluate to two values related under $\bool$.
  $v_1 = 0$ evaluates to $\false$ as $v_1$ is 1 and $\enot\; v_2$ evaluates to $\false$ as well as $v_2 = \true$.
  It remains to show $(\false,\false) \in \vprep[\extsub{ }{\alpha}{(\int,\bool,R)}]{\bool}$.
  For base types, the value interpretation relates equal values, so this is indeed true.
\end{proof}
\section{Recursive Types and Step Indexing}
\label{sec:rec-types}
Consider the following program from the \emph{untyped} lambda calculus:
\[
  \Omega = (\labs{x}{x \; x}) \; (\labs{x}{x \; x})
\]
The first bit of the evaluation of $\Omega$ is a $\beta$-reduction followed by a substitution at which point we end up with $\Omega$ again.
In other words, the evaluation of $\Omega$ diverges.
Moreover, it is not possible to assign a type to $\Omega$ (again the interested reader may try to verify this by attempting to assign a type).
It can hardly come as a surprise that it cannot be assigned a type, as we previously proved that the simply typed lambda calculus is strongly normalizing, so if we could assign $\Omega$ a type, then it would not diverge.

To type $\Omega$, we need recursive types.
If we are able to type $\Omega$, then we do not have strong normalization (as $\Omega$ is not strongly normalizing).
With recursive types, we can type structures that are inherently inductive such as lists, trees, and streams.
In an ML-like language, a declaration of a tree type would look like this:
\begin{lstlisting}
  type tree = Leaf
            | Node of int * tree * tree
\end{lstlisting}
In Java, we could define a tree class with an int field and two tree fields for the sub trees:
\begin{lstlisting}
  class Tree {
    int value;
    Tree left, right;
  }
\end{lstlisting}
In other words, we can define trees in common programming languages, but we cannot define them in the lambda calculus.
Let us try to find a reasonable definition for recursive types by considering what properties are necessary to define trees.
We want a type that can either be a node or a leaf.
A leaf can be represented by unit (in our trees leafs carry no information), and a node is the product of an int and trees.
We put the two constructs together with the sum type:
\[
  \tree = 1 + (\int * \tree * \tree)
\]
This is what we want, but the definition is not well-founded as it is.
We are defining the type $\tree$, but $\tree$ appears on the right hand side of the definition which makes it circular.
Instead of writing $\tree$, we use a type variable $\alpha$:
\begin{align*}
  \alpha &= 1 + (\int \times \alpha \times \alpha) \\
         &= 1 + (\int \times (\int \times \alpha \times \alpha) \times (\int \times \alpha \times \alpha)) \\
  &\vdots
\end{align*}
All sides of the above equations are equal, and they are all trees.
We could keep going and get an infinite system of equations.
If we keep substituting the definition of $\alpha$ for $\alpha$, we keep getting bigger and bigger types.
All of the types are trees, and all of them are finite.
If we take the limit of this process, then we end up with an infinite tree, and that tree is the tree we conceptually have in our minds.
So what we need is the fixed point of the above equation.

Let us define a recursive function for which we want to find a fixed point:
\[
  F = \lambda \alpha : type \ldotp1 + (\int \times \alpha \times \alpha)
\]
We want the fixed point which by definition is $t$ such that
\[
  t = F(t)
\]
So we want
\[
  tree = F(tree)
\]
The fixed point of this function is written:
\[
  \mu \alpha\ldotp F(\alpha)
\]
Here $\mu$ is a fixed-point type constructor. As the above is the fixed point by definition, it should be equal to $F$ applied to it:
\[
  \mu \alpha\ldotp F(\alpha) = F(\mu \alpha\ldotp F(\alpha))
\]
Now let us make this look a bit more like types by substituting $F(\alpha)$ for $\tau$. 
\[
  \mu \alpha\ldotp \tau = F(\mu \alpha\ldotp \tau) 
\]
The right hand side is really just $\tau$ with $\mu \alpha\ldotp \tau$ substituted with $\tau$:
\[
  \mu \alpha\ldotp \tau = \tau[\mu \alpha\ldotp \tau / \alpha]
\]
We introduce the recursive type $\mu \alpha\ldotp \tau$ to our language.
With the recursive type, we can shift our view to an expanded version $\tau[\mu \alpha\ldotp \tau / \alpha]$ and contract back to the original type.
Expanding the type is called $\unfold$ and contracting is is called $\fold$.
\[
\begin{tikzpicture}[->,>=stealth',shorten >=1pt,auto,node distance=3cm,
  thick,main node/.style={rectangle}]

  \node[main node] (1) {$\mu\alpha\ldotp \tau$};
  \node[main node] (2) [right of=1] {$\tau[\mu \alpha\ldotp \tau / \alpha]$};

  \path[every node/.style={font=\sffamily\small}]
    (1) edge [bend left] node [above] {$\unfold$} (2)
    (2) edge [bend left] node [below] {$\fold$} (1);
\end{tikzpicture}
\]
With recursive types in hand, we can now define our tree type:
\[
  \tree \eqdef \mu \alpha\ldotp 1 + (\int \times \alpha \times \alpha)
\]
When we work with a recursive type, we need to be able to open the type and use whatever is under the $\mu$.
Say we have an expression $e$ of type $\tree$.
We need to be able to tell whether it is a leaf or a node.
To do so, we need to peek under the $\mu$ which is done with an $\unfold$.
The $\unfold$ gives us an expression where the type has been unfolded by one level.
That is, the same expression but where $\alpha$ has been substituted with the definition of $\tree$ and the outer $\mu\alpha\ldotp$ has been removed.
With the outer $\mu\alpha\ldotp$ gone, we can match on the sum type to find out whether it is a leaf or a node.
When we are done working with the type, we can fold it to get the original tree type.
\begin{center}
  \begin{tikzpicture}[->,>=stealth',shorten >=1pt,auto,node distance=1.5cm,
      thick,main node/.style={rectangle}]

    \node[main node] (1) {$tree=\mu\alpha\ldotp 1+(\int \times \alpha \times \alpha)$};
    \node[main node] (2) [below of=1] {$1+(\int \times (\mu\alpha\ldotp 1+(\int \times \alpha \times \alpha)) \times (\mu\alpha\ldotp 1+(\int \times \alpha \times \alpha)))$};

    \path[every node/.style={font=\sffamily\small}]
    (1) edge [bend left=10] node [right] {$\unfold$} (2)
    (2) edge [bend left=10] node [left] {$\fold$} (1);
  \end{tikzpicture}
\end{center}
This kind of recursive types is called iso-recursive types because there is an isomorphism between a $\mu\alpha\ldotp \tau$ and its unfolding $\tau[\mu\alpha\ldotp \tau / \alpha]$. 

\subsection{Simply Typed Lambda Calculus with Recursive Types}
\label{subsec:stlc-rec-def}
The syntax of STLC with recursive types is defined as follows:
\begin{align*}
  \tau &::= \ldots \mid \mu \alpha\ldotp \tau \\
  e    &::= \ldots \mid \fold \; e \mid \unfold \; e \\
  v    &::= \ldots \mid \fold \; v\\
  E    &::= \ldots \mid \fold \; E \mid \; \unfold \; E
\end{align*}
Further, the following evaluation rule is added:
\begin{mathpar}
  \inferrule{ }{\unfold \; (\fold \; v) \evalto v}
\end{mathpar}
Finally, the following typing judgements are added:
\begin{mathpar}
  \TFold 
\and 
  \TUnfold
\end{mathpar}
This extension allows us to define a type for lists of integers:
\[
\int\; \listt \eqdef \mu\alpha\ldotp 1 + (\int \times \alpha)
\]
In the lambda calculus with recursive types, the $\Omega$ function has type $\tarrow{(\mu\alpha\ldotp \tarrow{\alpha}{\alpha})}{(\mu\alpha\ldotp \tarrow{\alpha}{\alpha})}$.
We do, however, need to rewrite it slightly as applies a function found under a $\mu$, so it must first unfold it:
\[
  \Omega = (\tlabs{x}{\mu\alpha\ldotp \tarrow{\alpha}{\alpha}}{(\unfold \; x) \; x}) 
\]
It is indeed possible to type this version of $\Omega$:
\[
  \inferrule*{\inferrule*{\inferrule*{\inferrule*{ }
                                                 {x : \mu\alpha\ldotp \tarrow{\alpha}{\alpha} \vdash x : \mu\alpha\ldotp \tarrow{\alpha}{\alpha} }}
                                     {x : \mu\alpha\ldotp \tarrow{\alpha}{\alpha} \vdash \unfold \; x : \tarrow{\tau_2}{(\mu\alpha\ldotp \tarrow{\alpha}{\alpha})}} \\
                         \inferrule*{\text{for $\tau_2 = \mu\alpha\ldotp \tarrow{\alpha}{\alpha}$} }
                                    { x : \mu\alpha\ldotp \tarrow{\alpha}{\alpha} \vdash x :\tau_2}}
                         {x : \mu\alpha\ldotp \tarrow{\alpha}{\alpha} \vdash (\unfold \; x) \; x : \mu\alpha\ldotp \tarrow{\alpha}{\alpha}}}
             {\mtenv \vdash \Omega : \tarrow{(\mu\alpha\ldotp \tarrow{\alpha}{\alpha})}{(\mu\alpha\ldotp \tarrow{\alpha}{\alpha})}}
\]

\subsection{A Step-Indexed Logical Predicate}
In a naive first attempt to construct the value interpretation, we could try:
\[
  \vpred{\mu\alpha\ldotp \tau} = \{\fold \; v \mid \unfold \; (\fold \; v) \in \epred{\sub{\tau}{\mu\alpha\ldotp\tau}{\alpha}} \}
\]
We can simplify this slightly; first we use the fact that $\unfold \; (\fold \; v)$ reduces to $v$.
Next we use the fact that $v$ must be a value and the fact that we want $v$ to be in the expression interpretation of $\tau[\mu \alpha\ldotp \tau / \alpha]$.
By unfolding the definition of the expression interpretation, we conclude that it suffices to require $v$ to be in the value interpretation of the same type.
With these simplifications we get the following definition:
\[
  \vpred{\mu\alpha\ldotp \tau} = \{\fold \; v \mid v \in \vpred{\sub{\tau}{\mu\alpha\ldotp\tau}{\alpha}} \}
\]
This is, however, not a well-founded definition.
The value interpretation is defined inductively on the type, but $\sub{\tau}{\mu\alpha\ldotp\tau}{\alpha}$ is not a structurally smaller type than $\mu\alpha\ldotp \tau$.

To solve this issue, we index the interpretation by a natural number, $k$, which we write as follows:
\[
  \vpres{\tau} = \{v \mid \dots \}
\]
Hence $v \in \vpres{\tau}$ is read as ``$v$ belongs to the interpretation of $\tau$ for $k$ steps.''
We interpret this in the following way: given a value that is part of an evaluation for $k$ or fewer steps (in the sense that the value is plugged into a program context, and the resulting expression evaluates for fewer than $k$ steps), then the value will have type $\tau$.
If we use the same value in a program context evaluates for more than $k$ steps, then we might notice that the value does not have type $\tau$ which means that we might get stuck.
This gives us an approximate guarantee.

We use the step as an inductive metric to make our definition well-founded.
That is we define the interpretation inductively on the step-index followed by an inner induction on the type structure.
Let us start by adding the step-index to our existing value interpretation:
\begin{align*}
  \vpres{\bool} &= \{\true,\false\} \\
  \vpres{\tarrow{\tau_1}{\tau_2}} &= \{\tlabs{x}{\tau_1}{e} \mid \forall j \leq k. 
  \; \forall v \in \vpres[j]{\tau_1}. \; \sub{e}{v}{x} \in \epres[j]{\tau_2} \}
\end{align*}
$\true$ and $\false$ are in the value interpretation of $\bool$ for any $k$, so $\true$ and $\false$ will for any $k$ look like it has type $\bool$.
To illustrate how to understand the value interpretation of $\tarrow{\tau_1}{\tau_2}$, please consider the following time line:  \\
\begin{center}
\begin{tikzpicture}
    % draw horizontal line   
    \draw[->] (0,0) -- (8,0);

    % draw vertical lines
    \foreach \x in {0,3,5,8}
      \draw (\x cm,3pt) -- (\x cm,-3pt);

    % draw nodes
    \draw (-2,0) node[below=3pt] { } node[above=30pt] { \footnotesize{$\lambda$ time-line} };
    \draw (0,0) node[below=3pt] {$k$} node[above=3pt] {\footnotesize{$(\tlabs{x}{\tau_1}{e}) \; e_2$}};
    \draw (3,0) node[below=3pt] {$ j+1 $} node[above=3pt] {\footnotesize{$(\tlabs{x}{\tau_1}{e})\; v $}};
    \draw (4.3,0) node[below=3pt] {$   $} node[above=3pt] {\footnotesize{$ \evalto $}};    
    \draw (5,0) node[below=3pt] {$ j $} node[above=3pt] {\footnotesize{$ \sub{e}{v}{x} $}};
    \draw (8,0) node[below=3pt] {$ 0 $} node[above=3pt] {$  $};
    \draw (9,0) node[below=3pt] { } node[above=3pt] { \footnotesize{"future"} };
  \end{tikzpicture}
\end{center}
Here we start at index $k$ and as we run the program, we use up steps until we at some point reach $0$ and run out of steps.
Say we at step $k$, we have an application of a lambda abstraction.
The lambda abstraction is already ready to be applied, but the application may not happen right away.
The $\beta$-reduction happens when the argument is an value, but the application may contain a non-value expression as it argument, i.e.\ $(\tlabs{x}{\tau_1}{e})\; e_2$.
It takes a number of steps to reduce $e_2$ to a value, and we have no way to tell how many it will take.
All we can do is to assume that it has taken some number of steps to evaluate $e_2$ to $v$, so we have $j+1$ steps left.
At this time, we can perform the $\beta$-reduction which means that we have $\sub{e}{v}{x}$ with $j$ steps left.

% If we ever hit 0 steps, then all bets are off. the value can have any type.

We can now define the value interpretation of $\mu\alpha\ldotp \tau$:
\[
  \vpres{\mu\alpha\ldotp \tau} = \{\fold \; v \mid \forall j < k. \; v \in \vpres[j]{\sub{\tau}{\mu\alpha\ldotp\tau}{\alpha}} \}
\]
This definition is almost the same as the one we proposed above but step-indexed.
In order to make the definition well-founded, we require $j$ be \emph{strictly} less than $k$.
We do not define a value interpretation for type variables $\alpha$, as we have no polymorphism yet.
The only place we have a type variable at the moment is in $\mu\alpha\ldotp \tau$, but in the interpretation we immediately close off the $\tau$ under the $\mu$, so we will never encounter a free type variable.

Finally, we define the expression interpretation:
\[
  \epres{\tau} = \{e \mid \forall j < k. \; \forall e'. \; e \evaltos[j] e' \wedge \irred(e') \; \implies \; e' \in \vpres[k-j]{\tau}\}
\]
To illustrate what is going on here, consider the following time line: \\
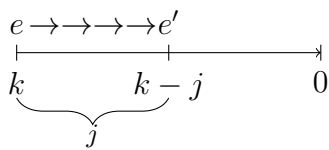
\begin{center}
\begin{tikzpicture}
    % draw horizontal line   
    \draw[->] (0,0) -- (4,0);

    % draw vertical lines
    \foreach \x in {0,2,4}
      \draw (\x cm,3pt) -- (\x cm,-3pt);

    % draw nodes
    \draw (0,0) node[below=3pt] {$k$} node[above=3pt] {$e$};
    \draw (1,0) node[below=3pt] {$ $} node[above=3pt] {$\evalto \evalto \evalto \evalto$};
    \draw (2,0) node[below=3pt] {$k-j$} node[above=3pt] {$e'$};
    \draw (4,0) node[below=3pt] {$0$} node[above=3pt] {$  $};

    % brace
    \draw [decorate,decoration={brace,amplitude=10pt,mirror}]
    (0,-0.6) -- (2,-0.6) node [black,midway,yshift=-0.5cm] 
          {\footnotesize $j$};
\end{tikzpicture}
\end{center}
We start with an expression $e$, then we take $j$ steps and get to expression $e'$.
At this point, if $e'$ is irreducible, then it must be in the value interpretation of $\tau$ for $k-j$ steps.
As explained above, the step-index approximates where we only can be sure that a value has a given type if we still have steps left.
In other words, if we allow the expression interpretation to spend all the steps, then we cannot say anything meaningful about the value.
% We use a strict inequality because we do not want to hit 0 steps.
% If we hit 0 steps, then we do not have any computational steps to observe a difference, so all bets are off.

We also need to lift the interpretation of type environments to step-indexing:
\begin{align*}
  \gpres{\mtenv} & = \{\emptyset \} \\
  \gpres{\Gamma, x : \tau} & = \{ \gamma[x \mapsto v] \mid \gamma \in \gpres{\Gamma} \wedge v \in \vpres{\tau} \}
\end{align*}
We are now in a position to lift the definition of semantic type safety to one with step-indexing.
\[
  \Gamma \models e : \tau \eqdef \forall k \geq 0. \; \forall \gamma \in \gpres{\Gamma} \ldotp \gamma(e) \in \epres{\tau}
\]
To actually prove type safety, we do it in two steps. First we state and prove the fundamental theorem:
\begin{theorem}[Fundamental property]
\label{thm:rec-ftlr}
 ~\\
  If $\Gamma \vdash e : \tau$, then $\Gamma \models e : \tau$.
\end{theorem}
When we have proven the fundamental property, we prove that it entails type safety, i.e.\
\[
 \mtenv \models e : \tau \implies \safe(e)
\]
Thanks to the way we defined the logical predicate, this second step should be trivial to prove.

The difficult part is to prove the fundamental property.
This proof requires a lemma that says the value interpretation is monotone, i.e.\ if a value is in the value interpretation for some step, then it is also in there for any smaller step.
\begin{lemma}[Monotonicity] ~\\
  \label{lem:rec-mono}
  If $v\in \vpres{\tau}$ and $j \leq k$,  then $v \in \vpres[j]{\tau}$.
\end{lemma}
\begin{proof}
The proof is by case on $\tau$.
\case{$\tau = \bool$}
assume $v \in \vpres{\bool}$ and $j \leq k$, we then need to show $v \in \vpres[j]{\bool}$. As $v \in \vpres{\bool}$, we know that either $v= \true$ or $v=\false$. If we assume $v=\true$, then we immediately get what we want to show, as $\true$ is in $\vpres[j]{\bool}$ for any $j$. Likewise for the case $v=\false$.
\case{$\tau = \tarrow{\tau_1}{\tau_2}$}
assume $v \in \vpres{\tarrow{\tau_1}{\tau_2}}$ and $j \leq k$, we then need to show $v \in \vpres[j]{\tarrow{\tau_1}{\tau_2}}$. As $v$ is a member of $\vpres{\tarrow{\tau_1}{\tau_2}}$, we can conclude that $v = \tlabs{x}{\tau_1}{e}$ for some $e$. By definition of $v \in \vpres[j]{\tarrow{\tau_1}{\tau_2}}$ we need to show
\[
  \forall i \leq j. \forall v' \in \vpres[i]{\tau_1}.\; \sub{e}{v'}{x} \in \epres[i]{\tau_2}
\]
Suppose $i \leq j$ and $v' \in \vpres[i]{\tau_1}$, we then need to show $\sub{e}{v'}{x} \in \epres[i]{\tau_2}$.

By assumption, we have $v \in \vpres{\tarrow{\tau_1}{\tau_2}}$ which gives us
\[
  \forall n \leq k. \; \forall v' \in \vpres[n]{\tau_1}.\; \sub{e}{v'}{x} \in \epres[n]{\tau_2}
\]
By transitivity, $j \leq k$, and $i \leq j$, we get $i \leq k$.
We use this with $v' \in \vpres[i]{\tau_1}$ to get $\sub{e}{v'}{x} \in \epres[i]{\tau_2}$ which is what we needed to show.
\case{$\tau = \mu \alpha\ldotp x$} assume $v \in \vpres{\mu\alpha\ldotp \tau}$ and $j \leq k$, we then need to show $v \in \vpres[j]{\mu\alpha\ldotp \tau}$.
From $v \in \vpres[k]{\tau}$, we conclude that there must exist a $v'$ such that $v = \fold \; v'$.
Now assume $i<j$ and show $v' \in \vpres[i]{\subst{\tau}{\mu\alpha\ldotp \tau}{\alpha}}$.
From $i<j$ and $j \leq k$, we can conclude $i < k$.
We use with
\[
  \forall n < k.\; v' \in \vpres[n]{\subst{\tau}{\mu\alpha\ldotp \tau}{\alpha}},
\]
which we get from $v \in \vpres{\mu\alpha\ldotp \tau}$, to get $v' \in \vpres[i]{\subst{\tau}{\mu\alpha\ldotp \tau}{\alpha}}$.
\end{proof}
\begin{comment}
  \begin{lemma}[Substitution]
    Let $e$ be syntactically well-formed term, let $v$ be a closed value and let $\gamma$ be a substitution that map term variables to closed values, and let $x$ be a variable not in the domain of $\gamma$, then
    \[
    \extsub{\gamma}{x}{v}(e) = \subst{\gamma(e)}{v}{x}
    \qedhere
    \]
  \end{lemma}
\end{comment}
\begin{proof}[Proof (Fundamental Property, Theorem~\ref{thm:rec-ftlr})]
Proof by induction on the typing derivation.
\case{\texttt{T-Fold}} \\~
Assume
\[
 \Gamma \vdash \fold \; e : \mu\alpha. \; \tau
\]
We need to show 
\newcommand{\mat}{\ensuremath{\mu\alpha\ldotp\tau}}
\[
  \Gamma \models \fold \; e : \mat
\]
So suppose we have $k \geq 0$ and $\gamma \in \gpres{\mat}$, then we need to show $\gamma(\fold\; e) \in \epres{\mat}$ which amounts to showing $\fold\; \gamma(e) \in \epres{\mat}$.

So suppose that $j<k$ and that $\fold\; \gamma(e) \evaltos[j] e'$ and $\irred(e')$, then we need to show $e' \in \vpres[k-j]{\mat}$.
As we have assumed that $\fold\; \gamma(e)$ reduces to something irreducible, and the operational semantics of this language are deterministic, we know that $\gamma(e)$ must have evaluated down to something irreducible.
We therefore know that $\gamma(e) \evaltos[j_1] e_1$ where $j_1 \leq j$ and $\irred(e_1)$.
Now we use our induction hypothesis: \newcommand{\tsub}{\ensuremath{\sub{\tau}{\mat}{\alpha}}}
\[
  \Gamma \models e : \tsub
\]
We instantiate this with $k$ and $\gamma \in \gpres{\Gamma}$ to get $\gamma(e) \in \epres{\tsub}$.
Which we then can instantiate with $j_1$ and $e_1$ to get $e_1 \in \vpres[k-j_1]{\tsub}$.

Now let us take a step back and see what happened: 
\[
\begin{array}{rcl}
  \fold\; \gamma(e) & \evaltos[j_1]& \fold\; e_1 \\
                    & \equiv{}& \fold\;v_1 \\
                    & \equiv{}& e'
\end{array}
\]
We started with a $\fold\; \gamma(e)$ which took $j_1$ steps to $\fold\; e_1$.
We have just shown that this $e_1$ is actually a value because it is in the value interpretation of $\vpres[k-j_1]{\tsub}$.
To remind us $e_1$ is a value let us henceforth refer to it as $v_1$.
We further know that $\fold\; \gamma(e)$ reduces to $e'$ in $j$ steps and that $e'$ is irreducible.
We can further conclude that $e' = \fold\; v_1$ and $j = j_1$ as the language is deterministic and $\fold\; v_1$ is irreducible (because it is a value).
Our proof obligation is to show $e' = \fold \; v_1 \in \vpres[k-j]{\mat}$ to show this suppose we have $l < k-j$ (this also gives us $l < k-j_1$ as $j = j_1$).
We then need to show $v_1 \in \vpres[l]{\tsub}$, we obtain this result from the monotonicity lemma using the assumption $v_1 \in \vpres[k-j_1]{\tsub}$ and $l < k-j_1$.
\end{proof}

The $\listt$ type from Section~\ref{subsec:stlc-rec-def} uses the sum type.
Sums are a straight forward extension to this language.
The extension of the value interpretation would be:
\[
  \vpres{\tau_1 + \tau_2} = \{\inl \; v_1 \mid v_1 \in \vpres{\tau_1}\} \cup 
                            \{\inr \; v_2 \mid v_2 \in \vpres{\tau_2}\}
\]
We can use $k$ directly or $k$ decremented by one depending on whether casing should take up a step or not.
Either way the definition is well-founded.

\subsection{Exercises}
\begin{enumerate}
\item Do the lambda and application case of the proof of the \emph{Fundamental Property} (Theorem~\ref{thm:rec-ftlr}).%Specify in what proof, probably fundemental property
\item Try to prove the \emph{monotonicity} lemma where the definition of the value interpretation has been adjusted with:
\[
\vpres{\tarrow{\tau_1}{\tau_2}} = \{\tlabs{x}{\tau_1}{e} \mid \forall v \in \vpres{\tau_1}. \; \sub{e}{v}{x} \in \epres{\tau_2} \}
\]
This will fail, but it is instructive to see how it fails.
\end{enumerate}
\section{References and Worlds}
\label{sec:references}
This section is not based on Amal Ahmed's lectures at OPLSS '15, so the reader will have to make do with these notes.

\subsection{\STLC{} with References}
In order to add references to \STLC{}, we add means to allocate new references, make assignments to existing references, and dereference existing references:
\begin{align*}
  \tau & ::= \dots \mid \Tref{\tau} \\
  e    & ::= \dots \mid \Ealloc{e} \mid \Eassign{e}{e} \mid \Ederef{e} \mid l\\
  E    & ::= \dots \mid \Ealloc{E} \mid \Eassign{E}{e} \mid \Eassign{v}{E} \mid \Ederef{E} \\
  v    & ::= \dots \mid l
\end{align*}
The expression $\Ealloc{e}$ allocates a new cell with an initial value specified by $e$.
The expression $\Eassign{e}{e}$ makes an assignment to an existing reference.
The expression $\Ederef{e}$ dereferences an existing reference.
We also add locations $l$ to the language.
Locations are references to the heap.

In order to model references, we need to add a store or a heap to our language.
Our heap $h$ is simply a finite partial map from location $\Loc$ to the values of the language:
\[
  h \; : \; \Loc \finfp \Val
\]
This is a reasonable model of a real heap: The map contains all the allocated locations.
Any real program will only ever allocated a finite amount of memory, and we have infinitely many available heap cells, so we do not have to worry about running out of memory.

We also need to update the operational semantics, so it has the heap available during evaluation.
We update the step relation to be a partial relation on configurations where a configuration is a pair of a heap and an expression.
\[
  \tuple{h,e} \evalto \tuple{h',e'}
\]
The pure reductions of the language are all the reductions of \STLC{}.
These reductions are pure because they do not make changes to the heap.
% Pure reduction like normal
\begin{mathpar}
  \inferrule{ e \evalto e'}
            { \tuple{h,E[e]} \evalto \tuple{h,E[e']} }
\end{mathpar}
We introduce an impure reduction, i.e.\ a reduction that manipulate or interacts with the heap, for each of the three new ways to interact with the heap.
% Impure reductions
\begin{mathpar}
  \inferrule*[right=\textsc{E-Alloc}]{ l \not\in \dom(h) }
             {\tuple{h,E[\Ealloc{v}]} \evalto \tuple{\extendh{h}{l}{v},E[l]}}
\and
  \inferrule*[right=\textsc{E-Assign}]{ l \in \dom(h) }
             {\tuple{h,E[\Eassign{l}{v}]} \evalto \tuple{\extendh{h}{l}{v},E[v]}}
\and
  \inferrule*[right=\textsc{E-Deref}]{ l \in \dom(h) }
             { \tuple{h,E[\Ederef{l}]} \evalto \tuple{h,E[h(l)]} }
\end{mathpar}
An allocation $\Ealloc{v}$ allocates a new location $l$ on the heap with the initial value $v$.
The location must be new in the sense that the heap did not previously use that location.
Further, we do not expose the locations in the surface language (this is enforced by the type system), so allocation is the only way to obtain locations.
An assignment $\Eassign{l}{v}$ updates the heap at location $l$ to point to $v$.
Here the result of an assignment is the assigned value, but often it will just be unit value (we do not opt for this as we have not introduced unit value).
Finally, $\Ederef{l}$ dereferences location $l$ which means that it looks up the value denoted by $l$ in the heap.

With the heap introduced, we will also have to change some of the definitions that have been reoccurring in previous sections.
One such definition is $\irred$:
\[
  \begin{gathered}
    \irred(h,e)\\
    \text{iff}\\
    \nexists h',e' \ldotp \tuple{h,e} \evalto \tuple{h',e'}
  \end{gathered}
\]

For each of the new expressions, except locations, we introduce a new typing rule to the type system:
\begin{mathpar}
  \infer*[right=\textsc{T-Alloc}]{ \Gamma \vdash e : \tau}
        { \Gamma \vdash \Ealloc{e} : \Tref{\tau}}
\and
  \infer*[right=\textsc{T-Assign}]{ \Gamma \vdash e_1 : \Tref{\tau}\\
         \Gamma \vdash e_2 : \tau}
        { \Gamma \vdash \Eassign{e_1}{e_2} : \tau }
\and
  \infer*[right=\textsc{T-Deref}]{ \Gamma \vdash e : \Tref{\tau}}
        { \Gamma \vdash \Ederef{e} : \tau }
\end{mathpar}
% Do we mention why locations are not here?
% Not part of the surface language is one reason.

We also introduce the notion of a well-types heap:
\[
  \Gamma \vdash h : \Sigma \quad \text{iff} \quad \left\{
    \begin{array}{l}
      \dom(h) = \dom(\Sigma) \wedge\\
      \forall l \in \dom(h)\ldotp \Sigma; \Gamma \vdash h(l) : \Sigma(l)
    \end{array}\right.
\]

\subsection{Properties of \STLC{} with References}
With the language defined, we take a step back to consider what properties the language have.
\begin{description}
\item[The operational semantics is non-deterministic:]
  In the reduction rule \textsc{E-Alloc}, we only require $l$ to be a new location that is not in the domain of the heap.
  It can, however, be any new location which is the cause of the non-determinism.
  By leaving allocation underspecified, our system can handle different real implementations of allocation.
  One drawback of non-determinism is that we cannot rely on determinism in our proofs (which we have previously done).
\item[Evaluation can get stuck:] The reduction rules for assignment and dereference require the location to be in the heap.
  If this is not the case, then the evaluation is stuck.
  One way to think of this is as a memory fault.
  No surface language for expressing locations is provided, so to obtain a reference one has allocate it which means that there should (hopefully) be no way to obtain a reference for memory that has not been allocated.
\item[Only mutation on the heap:]
  The only place we have introduced mutation is on the heap.
  Variables bound with a $\lambda$-expression can still not be modified.
\item[Values in the same memory cell stay the same type:]
  The type system enforces the invariant that memory cells always store values of the same type.
\item[Recursion:] 
  It may come as a surprise that the language has recursion as it does not have recursive types or a fixed-point operator.
  However, recursion can be achieved through the heap with a technique known as \emph{Landin's knot}\citep{Landin:CJ64}.

We demonstrate this technique with the following program
\begin{lstlisting}[escapeinside={@}{@},basicstyle=\footnotesize\ttfamily]
(let x = ref (\ y : int. y) in
    x := (\ n : int. (!x 0));
    !x) 0
\end{lstlisting}
which recurses through the heap to diverge.
We have written the program in an ML-style, but it can be written as a well-typed expression (see Appendix~\ref{app:landins-knot}) in the language we have presented here (assuming the trivial extension with integers).
The program first allocates a new reference with a dummy value as the default value.
Specifically, the dummy is chosen, so it has the type of the value we actually want to store (in this case, the dummy is the identity function on integers), but otherwise the function does not matter.
The reference is stored in variable \texttt{x}.
Next, the program assigns the function we are actually interested in to \texttt{x}.
This function dereferences \texttt{x} to get a function which it applies to $0$.
The function takes an integer as an argument in order to have the correct type.
At the point when this function is actually applied, the location denoted by \texttt{x} will contain the function itself, so it will call itself repeatedly.
To start the recursion, the program dereferences \texttt{x}, which yields the function we just described, and applies the result to 0.

The example should give an idea about how Landin's knot can be used to emulate a fixed-point operator which gives general recursion.
\end{description}

\subsection{Logical Predicate}
In this section, we define a logical relation that we use to prove type safety of \STLC{} with references.

We will have to deal with a couple of new things, but the main thing is how to interpret locations in the value interpretation.
The interpretation will of course look something like this
\[
  \vpred{\Tref{\tau}} = \{l \mid \dots \},
\]
but what should we require from the location?
To answer this question, we should consider the elimination form of things of reference type which is dereference.
Dereferencing a location yields the value stored in that location of the heap, so that value should be in the value interpretation, i.e.\ $h(l) \in \vpred{\tau}$.
However, the value interpretation does not consider a heap in its interpretation, so we cannot write this immediately.
We could try to index the value interpretation with a heap, but we are not just interested in safety with respect to one heap.
Instead we need a semantic model of the heap.

A semantic model for heaps is called a world.
It specifies what values one can expect to find at a given location.
We will define worlds in detail later, but for now we will just assume that a world is a function from locations to value predicates.
We index our value interpretation with a world, i.e.\ 
\[
  \vpred{\Tref{\tau}}(W) = \{l \mid \dots \},
\]
Returning to the elimination form, we still cannot express that the contents of the heap at a specific location has to be in the value interpretation.
We can, however, put restrictions on what the world can allow to be at a specific position.
Specifically, we want to make sure that if a location $l$ is in the value interpretation, then $l$ can only dereference values that respect safety.
This can be achieved by requiring the world to only allow values that respect safety to reside at location $l$.
The values that respect safety are exactly what is captured by the value interpretation, so we require the predicate at location $l$ in the world to be equal to the value interpretation.
\[
  \vpred{\Tref{\tau}}(W) = \{l \mid l \in \dom(W) \wedge W(l) = \vpred{\tau} \},
\]
This is exactly the definition we would like for our value interpretation of reference types, but unfortunately we cannot use it.
The problem is that the domain of the worlds cannot exist.
To see why, we need to take a step back and consider what we have been doing when defining value interpretations so far.
The value interpretations defined in previous sections have been elements in the space of all predicates over values.
We call such a space \emph{the space of semantic types}.
For instance in Section~\ref{sec:stlc-type-safety}, the space of semantic types $T$ was $\Pred{\Val}$.
The value interpretation we want now is world indexed, so it needs to be from a different space of semantic types, namely
\[
  T = \World \fun \Pred{\Val}
\]
A world is a function from locations to safe values.
A predicate of safe values is in the space of semantic types, so the domain of worlds is
\[
  \World = \Loc \finfp T
\]
If we inline $T$ in the definition of $\World$, we get
\[
  \World = \Loc \finfp (\World \fun \Pred{\Val})
\]
which is a recursive domain equation for which no solution exists.
This means that we cannot have the proposed definition of the value interpretation as the worlds we need for the definition do not exist.
It may not come as a surprise that the proposed definition does not work.
After all, this language is similar in its attributes to \STLC{} with recursive types.
When we defined the logical relation for that language, we had to introduce step-indexing to approximate the interpretation of values.
In our proposed definition, we have nothing like step-indexing, and it would be a bit surprising if we could make do without here.

The solution to circularity in the domain equation seems intuitively similar to the one for recursive types.
Rather than using an exact solution to the recursive domain equation, we use an approximate solution.
Developing such a solution is beyond the scope of this note, so we only present the definitions necessary to understand the rest.
The following is the approximate solution:
\[
\xi : \hat{T} \cong \blater (( \Loc \finfp \hat{T}) \monfp \UPred{\Val})
\]
(ignore the $\blater$ for now) which allows us to define the domain of worlds
\[
  \World = \Loc \finfp \hat{T}
\]
and finally to define the space of semantic types as
\[
  T = \World \monfp \UPred{\Val}
\]

In order to have an approximation, we approximate over a metric.
In this case, the metric is steps.
This means that our predicates should be step-indexed with a natural number which is where $\mathrm{UPred}$ comes into the picture.
$\mathrm{UPred}$ defines the step-indexed, downwards closed predicates over some domain\footnote{$\UPred{\Val}$ was the space of semantic types from which we picked the value interpretation for the logical predicate for \STLC{} with recursive types}:
\[
  \UPred{\Val} = \{ A \subseteq  \nats \times \Val \mid \forall (n, v) \in A. \forall m \leq n\ldotp (m, v) \in A\}
\]

Because of the approximation, equalities only hold true as long as we have steps left.
When we run out of steps, we can no longer distinguish elements from each other and all bets are off.
Such an equality is an $k$-equality: $\nequal$.
An $k$-equality has properties reminiscent of step-indexed logical relation for \STLC{} with recursive types (\mulr{}): In \mulr{} the step-index was the amount of steps we could take without observing that a value was not of a certain type.
This meant that when we hit zero steps, everything would be in \mulr{}.
Similarly, for an $k$-equality $\nequal[0]$ is the total relation.
For \mulr{} we had a lemma that stated that the value interpretation was downwards closed.
Here we require that the $k$-equalities are downwards closed, i.e.\ $\nequal \subseteq \nequal[k+1]$.

We could not solve the original domain equation due to cardinality issues related to sets.
In some sense, $\nats \finfp T$ was simply too big.
In fact in order to solve the recursive domain equation, we move away from the familiar space of sets into a different space.
This space can be thought of as the space of sets with additional structure.
The structure is related to the step-indices, and everything in the space must preserve this structure.
We mention this because all the definitions that constitute our logical predicate are in this new space which means that they must preserve the added structure of the spaces.
In our proofs, we will also have to make sure this structure is preserved which we sometimes do without mention.
The proofs also use some of this structure without mention.
We will not prove that our definitions have the necessary structure, but we will leave it as an exercise in the end of this section.

There are still a couple of things we have not mentioned in the domain equation.
The first is that the space of semantic types are monotone functions with respect to the world.
Later in this section, we present the preorder the functions are monotone with respect to along with some intuition for why we want them to be monotone.
We will state this as a lemma later, but we will also note that all the predicates mapped to by a world have to be monotone with respect to worlds.
The second is the fact that the solution to the recursive domain equation is an isomorphism $\xi$.
This means that our definitions will use $\xi$ to move between the two sides.
This is also a technicality, but we include it so our definitions are correct.
The third and last thing is the black triangle, known as \emph{later}.
The purpose of the later in the recursive domain equation is to make that space of structured sets under it small enough for a solution to exist.
In practice, later makes all $k$-equalities under it go one step down.
The later is important for the construction of the solution for the recursive domain equation, but that is beyond the scope of this section (we refer the interested reader to~\citet{Birkedal:tutorial-notes}).

We can now define the $k$-equality for $\mathrm{UPred}$.
\begin{definition}
  For $B \in \UPred{A}$ and $k\in\nats$
  \[
    \cut{B} \eqdef \{(n,a)\in B | n < k \} \qedhere
  \]
\end{definition}
\begin{definition}
  For $B,C \in \UPred{A}$ and $k\in\nats$
  \[
    B \nequal C \quad \text{iff}\quad \cut{B} = \cut{C} \qedhere
  \]
\end{definition}
That is, two predicates are $k$-equal if they are equal on all elements with smaller steps than $k$.

We can now state the value interpretation for the reference type:
\[
  \vpred{\Tref{\tau}}(W) = \{(k,l) \mid l \in \dom(W) \wedge \xi(W(l)) \nequal[k] \vpred{\tau} \},
\]
This is almost the same as the definition we proposed first.
The only difference is that the predicate is step indexed (here we use an explicit step, but it is essentially the same as what we had for recursive types), and it uses a $k$-equality rather than a normal equality.

We also need to look at the remainder of the value interpretation.
As in all the other logical predicates and relations, the boolean values are always safe, so we just need to add a step index for it to conform with the remaining definitions
\[
  \vpred{\bool}(W) =\{ (k,\true), (k,\false) \mid k \in \nats\}
\]
For the function type, we can start with a similar definition to what we had for recursive types:
\[
  \vpred{\tarrow{\tau_1}{\tau_2}}(W) = \{ (k,\lambda  x\ldotp e) \mid \forall j\leq k \ldotp \forall (j,v) \in \vpred{\tau_1}(W) \ldotp (j,e[v/x]) \in \epred{\tau_2}(W)\}
\]
Recall that we consider all $j \le k$ because an application of a lambda-abstraction may not happen immediately.
Specifically, the argument may be an expression that needs to evaluate to a value before the application can take place.
In this language, evaluating an expression may have side effects.
It may assign new values to certain locations, which is fine as long as they are in the value interpretation, or it may also allocate new references which corresponds to adding new locations to the heap.
Because of the non-determinism of location allocation, we cannot statically (that is before the execution) say what type of values will be stored where on the heap, so we the world needs to be updated dynamically with this information.
To this end, we construct a future world relation that captures allocation of new locations.
\begin{definition}[Future worlds]
  For worlds $W$ and $W'$:
  \[
    W' \sqsupseteq W \quad \text{ iff } \dom(W') \supseteq \dom(W) \wedge \forall l \in \dom(W)\ldotp W'(l) = W(l)
  \]
  For $W' \sqsupseteq W$, we say \emph{$W'$ is a future world of $W$.}
\end{definition}
The future world relation is extensional in the sense that a future world remains the same as the past world with respect to everything but possible extensions.
We use this to finish our value interpretation by saying that the argument should be valid in a future world which means that it may have allocated new locations.
\[
  \vpred{\tarrow{\tau_1}{\tau_2}}(W) = \left\{ (k,\lambda x\ldotp e) \middle| 
    \begin{multlined}
    \forall W' \sqsupseteq W, j\leq k, (j,v) \in \vpred{\tau_1}(W') \ldotp\\ (j,e[v/x]) \in \epred{\tau_2}(W')
  \end{multlined}
  \right\}
\]
Generally speaking, if a value is safe with respect to a world, then it should remain that even if the presence of changes to the heap.
For instance, if we have a safe location $l$, i.e.\ a location that is in the value interpretation, then its safety should not be affected by allocation of new locations.
We should still only be able to use $l$ to dereference safe values.
As the future world relation models changes to the heap, we want our value interpretation to be monotone with respect to future worlds, so safe values are resilient to allocation of new references.

The expression interpretation is also going to look like the one from the logical predicate for \STLC{} with recursive types.
That is, if an expression reduces to an irreducible expression, then it must be in the value interpretation.
However in this language, it is configurations and not expressions that take evaluation steps, so we need to specify with respect to what heap the expression evaluates in.
To this end, we simply take any heap that satisfies the world.
For a heap to satisfy a world, it should have the locations specified by the world; and for all the locations in the heap, it should contain a value permitted by the world.
\begin{definition}[Heap satisfaction]
For a heap $h$ and world $W$, we have
\[
  \hsat{h}{W} \text{ iff } \dom(h) = \dom(W) \wedge \forall l \in \dom(h)\ldotp (k,h(l)) \in \xi(W(l))(W)
\]
When $\hsat{h}{W}$ we say: \emph{$h$ ($k$-)satisfies $W$}.
\end{definition}
The future world relation is a partial order (exercise).
With this definition, we are ready to state the expression interpretation:
\[
  (k,e)\in\epred{\tau}(W) \text{ iff } \left\{
    \begin{multlined}
    \forall j \le k, i < j, h, h', e', W' \sqsupseteq W \ldotp\\ \phantom{space}h :_j W' \wedge \tuple{h,e} \evaltos[i] \tuple{h',e'} \wedge \irred(h',e') \\ \implies\\ \exists W'' \sqsupseteq W' \ldotp h' :_{j-i} W'' \wedge (j-i,e') \in \vpred{\tau}(W'')
  \end{multlined}
\right.
\]
Note that for the result heap $h'$ there should exists some future world of $W$ such that $h'$ satisfies that world.
This means that during the evaluation the expression can only have made allocations consistent with the heap it started evaluation in.

As per usual, we need to have an interpretation of typing contexts.
It is defined in a straight forward manner by lifting the simple definition with worlds and step-indexing:
\begin{align*}
  \gpred{\emptyset}(W)     & ::= \{(k,\emptyset)\} \\
  \gpred{\Gamma,x:\tau}(W) & ::= \{(k,\gamma[x \mapsto v]) \mid (k,v) \in \vpred{\tau}(W) \wedge (k,\gamma) \in \gpred{\Gamma}(W)\} \\
\end{align*}
Finally, we can state semantic type safety:
\[
  \Gamma \models e : \tau \quad\text{iff}\quad \forall W, k \geq 0, (k,\gamma) \in \gpred{\Gamma}(W)\ldotp(k,\gamma(e))\in\epred{\tau}(W)
\]

\subsection{Safety Proof}
The safety proof has the same structure as in the previous sections where we proved safety.
We do need a number of lemmas for the proof.
We leave the proofs of these lemmas as exercises.

First of all, we need our logical predicate to be monotone with respect to worlds.
\begin{lemma}[World monotonicity]
  \label{lem:world-monotonicity}
  For $W' \sqsupseteq W$ we have
  \begin{itemize}
  \item If $(k,v) \in \vpred{\tau}(W)$, then $(k,v) \in \vpred{\tau}(W')$
  \item If $(k,e) \in \epred{\tau}(W)$, then $(k,e) \in \epred{\tau}(W')$
  \item If $(k,\gamma) \in \gpred{\Gamma}(W)$, then $(k,\gamma) \in \gpred{\Gamma}(W')$
    \qedhere
  \end{itemize}
\end{lemma}
Note that it is not all of our definitions that are monotone with respect to future worlds.
Specifically, heap satisfaction is not monotone because it would be nonsensical.
If heap satisfaction was monotone, then a heap would have to be able to handle that worlds that requires more locations allocated than it has.

We also need to make sure that our definitions are downwards closed.
This is reminiscent of Lemma~\ref{lem:rec-mono} from Section~\ref{sec:rec-types} (we do not call it monotonicity here to not confuse it with world monotonicity).
For instance, if a value is in the value interpretation at index $k$, then it should also be in there for any smaller step.
\begin{lemma}[Downwards closure]
  \label{lem:ref-downwards-closed}
  For $j \leq k$
  \begin{itemize}
  \item If $(k,v) \in \vpred{\tau}(W)$, then $(j,v) \in \vpred{\tau}(W)$
  \item If $(k,e) \in \epred{\tau}(W)$, then $(j,e) \in \epred{\tau}(W)$
  \item If $(k,\gamma) \in \gpred{\Gamma}(W)$, then $(j,\gamma) \in \gpred{\Gamma}(W)$
  \item If $\hsat{h}{W}$, then $\hsat[j]{h}{W}$
\qedhere
  \end{itemize}
\end{lemma}
Finally, we need a substitution lemma for the \texttt{T-Abs} case for the Fundamental Property proof.
\begin{lemma}
  \label{lem:ref-substitution}
  Let $e$ be syntactically well-formed, $v$ a closed value, $\gamma$ a substitution of closed values where $x$ is not mapped
  \[
    \gamma(\subst{e}{v}{x}) \equiv \gamma[x \mapsto v](e)
  \]
\end{lemma}
We can now state and prove the fundamental property.
\begin{theorem}[Fundamental Property]
  \label{thm:ref-ftlr}
  \[
    \Gamma \vdash e : \tau \implies \Gamma \models e : \tau
    \qedhere
  \]  
\end{theorem}
\begin{proof}
  By induction over the typing derivation.
  In all cases, the case for $k=0$ is vacuously true as no $j<0$ exists.
 \case{\textsc{T-True}} Let $W$, $k > 0$, and $(k,\gamma) \in \gpred{\Gamma}(W)$ be given and show:
  \[
    (k,\true)\in\epred{\bool}(W)
  \]
  To this end let $j < k$, $h$, $h'$, and $e'$ be given such that
  \begin{itemize}
  \item $\hsat[j]{h}{W}$,
  \item $\tuple{h,\true}\evaltos[j] \tuple{h',e'}$, and
  \item $\irred(h',e')$.
  \end{itemize}
  By the evaluation relation, it must be the case that $j=0$, $h'=h$ and $e'=\true$.
  Now pick $W' = W$.
  By assumption we have $\hsat[j]{h}{W}$ which is one of the two things we need to show.
  The other thing we must show is $(k,\true)\in\epred{\bool}(W)$.
  This follows immediately from the definition of the value interpretation as boolean values are always in there.
  \case{\textsc{T-False}} Analogous to the case for \textsc{T-True}.
  \case{\textsc{T-Var}} Also analogous to the case for \textsc{T-True} but to argue $(k,\gamma(x))\in\epred{\tau}(W)$, we use assumptions $(k,\gamma) \in \gpred{\Gamma}(W)$ and $\Gamma \vdash x : \tau$.

~
  \case{\textsc{T-Deref}}
  Let $W$, $k > 0$, and $(k,\gamma) \in \gpred{\Gamma}(W)$ be given and show:
  \[
    (k,\Ederef{\gamma(e)})\in\epred{\tau}(W)
  \]
  To this end let $j < k$, $h$, $h'$, and $e'$ be given such that
  \begin{itemize}
  \item $\hsat[j]{h}{W}$,
  \item $\tuple{h,\Ederef{\gamma(e)}}\evaltos[j] \tuple{h',e'}$, and
  \item $\irred(h',e')$.
  \end{itemize}
  By the assumed evaluation, either the dereferenced expression $\gamma{e}$ gets stuck, or it evaluates to some value, i.e.\ 
  \begin{enumerate}
  \item \label{item:ref-ftlr-deref-stuck} $\tuple{h,\gamma(e)} \evaltos[j] \tuple{h'',e''}$ and $\irred(h'',e'')$ for some $h''$ and non-value $e''$.
  \item \label{item:ref-ftlr-deref-val} $\tuple{h,\gamma(e)} \evaltos[i] \tuple{h'',v}$ and for some $h''$, value $v$, and $i<j$.
  \end{enumerate}
  In case~\ref{item:ref-ftlr-deref-stuck}, we use our induction hypothesis: $\Gamma \models e : \Tref{\tau}$.
  From this we conclude that $(k-j,e'') \in \vpred{\Tref{\tau}}(W')$ for some $W'$ which means that $e''$ is a value contradicting that it is a non-value.

  In case~\ref{item:ref-ftlr-deref-val}, we use the same induction hypothesis to get
  \begin{itemize}
  \item $\hsat[k-i]{h''}{W''}$ and
  \item $(k-i,v) \in \vpred{\Tref{\tau}}(W'')$
  \end{itemize}
  for some $W'' \future W$.
  By the definition of the value interpretation this means that $v = l$ for some location $l$.
  Now the dereference can take place, i.e.
  \[
    \tuple{h'',l} \evalto \tuple{h'',h''(l)}
  \]
  The heap only contains values, so $\tuple{h'',h''(l)}$ must be irreducible.
  Since we consider the assumed evaluation it must be the case that $h' = h''$ and $e' = h'(l)$.
  We now want to pick the world necessary for the expression relation.
  Our pick needs to satisfy $\hsat[k-j]{h'}{W'}$ and $W' \future W$.
  Since the dereference did not change the heap and we assumed heap satisfaction for $h''$ under $W''$, we can simply pick $W' = W''$.
  However, the assumed heap satisfaction is for step $k-i$, but we need it for the smaller step $k-j$, i.e. $\hsat[k-j]{h'}{W'}$.
  Generally when need something for a specific step, it suffices to show it for a greater step because the result will follow from Lemma~\ref{lem:ref-downwards-closed} which it also does in this case.

  It remains to show $(k-j,h'(l)) \in \vpred{\tau}(W')$.
  That is, the value dereferenced from the heap is safe. 
  To show this we need the following:
  \begin{itemize}
  \item From $\hsat[k-i]{h'}{W'}$, we get that $(k-i,h'(l)) \in \xi(W'(l))$.
  \item From $(k-i,l) \in \vpred{\Tref{\tau}}(W')$, we get $\xi(W'(l)) \nequal[k-i] \vpred{\tau}$.
  \end{itemize}
  From this, we conclude $(k-i-1,h'(l)) \in \vpred{\tau}(W)$ (remember that the definition of $k$-equality is defined in terms of $k$-cut which takes away a step).
  We know that $i<j$, so $k-j \le k-i-1$, so the desired result follows from Lemma~\ref{lem:ref-downwards-closed}.
 \case{\textsc{T-Alloc}} 
  Let $W$, $k > 0$, and $(k,\gamma) \in \gpred{\Gamma}(W)$ be given and show:
  \[
    (k,\Ealloc{\gamma(e)})\in\epred{\Tref{\tau}}(W)
  \]
  To this end let $j < k$, $h$, $h'$, and $e'$ be given such that 
  \begin{itemize}
  \item $\hsat[j]{h}{W}$ and
  \item $\tuple{h,\Ederef{\gamma(e)}} \evaltos[j] \tuple{h',e'}$, and
  \item $\irred(h',e')$.
  \end{itemize}
  By the assumed evaluation one of the following must be the case
  \begin{enumerate}
  \item \label{item:ref-ftlr-alloc-stuck} $\tuple{h,\gamma(e)} \evaltos[j] \tuple{h'',e''}$ and $\irred(h'',e'')$ for some $h''$ and non-value $e''$.
  \item \label{item:ref-ftlr-alloc-val} $\tuple{h,\gamma(e)} \evaltos[i] \tuple{h'',v}$ and for some $h''$, value $v$, and $i<j$.
  \end{enumerate}
  In case~\ref{item:ref-ftlr-alloc-stuck}, we use our induction hypothesis: $\Gamma \models e : \tau$.
  This gives us that $(k-j,e'') \in \vpred{\tau}(W')$ for some $W'$ which means that $e''$ is a value contradicting that it is not.

  In case~\ref{item:ref-ftlr-alloc-val}, we use the same induction hypothesis to get
  \begin{itemize}
  \item $\hsat[k-i]{h''}{W''}$ and
  \item $(k-i,v) \in \vpred{\Tref{\tau}}(W'')$
  \end{itemize}
  for some $W'' \future W$.
  Now that $e$ has been evaluated to a value, the allocation can happen.
  Therefore, the next evaluation step is
  \[
    \tuple{h'',v} \evalto \tuple{h''[l \mapsto v], l}\text{ for $l\notin\dom(h'')$}
  \]
  As we have just been looking at the assumed evaluation and this expression is irreducible, it must be the case that $e'=l$ and $h'=h''[l\mapsto v]$.
  Now we need to pick the world for the expression interpretation.
  In the last step of the evaluation, a new location has been allocated on the heap, so we cannot pick the assumed world $W''$ as it does not mention this location.
  We do, however, use it as a basis for our world and pick $W' = W''[l \mapsto \xi^{-1}(\vpred{\tau})]$, i.e.\ we allow anything from the value interpretation of $\tau$ to reside at location $l$ (note that we need to apply $\xi^{-1}$ for technical reasons - it can be safely ignored).
  This means that we need to show the following
  \begin{itemize}
  \item $W' \future W$

    By transitivity of $\future$ it suffices to show $W' \future W''$ as we already know $W'' \future W$.
    Intuitively, this follows from the fact that $W'$ extends $W''$ with $l$.
    We know that $l$ is not in the domain of $W''$ by assumptions $l\notin\dom(h'')$ and $\hsat[k-i]{h''}{W''}$.

  \item $\hsat[k-j]{h''[l \mapsto v]}{W'}$

    By the above argument, we know that $l$ is not in the domain of $h''$ and that $\dom(h'') = \dom(W'')$, so from the way we defined $W'$ it easily follows that $\dom(h''[l \mapsto v]) = \dom(W')$.

    Take $l'\in\dom(W')$. We need to consider two cases:
    \begin{itemize}
    \item $l' = l$: In this case we must show $(k-j,h''[l \mapsto v](l)) \in \xi(W'(l))(W')$, i.e.\ $(k-j,v) \in \vpred{\tau}(W')$, which follows by assumption and Lemma~\ref{lem:ref-downwards-closed}.

    \item $l' \neq l$: This follows by $\hsat[k-i]{h''}{W''}$ and the fact that $\xi(W''(l))$ is monotone with respect to the world (to see that this must be the case, take a look at the recursive domain equation).
    \end{itemize}
  \item $(k-i-1,l) \in \vpred{\Tref{\tau}}(W')$

    This amounts to showing $\xi(W(l)) = \xi(\xi^{-1}(\vpred{\tau})) \nequal[k-i-1] \vpred{\tau}$ which is true as if they are equal, then they are also equal if we limit them to everything of index less than $k-i-1$.
  \end{itemize}
~
\end{proof}
We leave the remaining cases of the above proof as exercises.

The safety predicate we used in previous sections needs to change a bit in this setting, so it takes the heap into account
\[
  \begin{gathered}
    \safe(e)\\
    \text{iff}\\
    \forall h', e' \ldotp \tuple{\emptyset,e} \evaltos \tuple{h',e'} \implies \val(e') \vee \exists h'', e''\ldotp \tuple{h',e'} \evalto \tuple{h'',e''}
  \end{gathered}
\] 
We now show the standard ``second lemma'' that says that the logical predicate is adequate to show safety.
\begin{lemma}
  \label{lem:ref-adeq}
  \[
    \emptyset \models e : \tau \implies \safe(e)
  \]
\end{lemma}
\begin{proof}
  Assume $\emptyset \models e : \tau$ and let $h$, $h'$, $e'$ be given such that $\tuple{h,e} \evaltos \tuple{h',e'}$.
  Proceed by case on $\irred(h',e')$.
  \case{$\neg\irred(h',e')$}: In this case, it follows by definition of $\irred$ that there exists $h''$ and $e''$ such that $\tuple{h',e'} \evalto \tuple{h'',e''}$.
  \case{$\irred(h',e')$}: Say the evaluation takes $k$ steps to do the assumed evaluation, i.e.\ $\tuple{\emptyset,e} \evaltos[k] \tuple{h',e'}$.
  Now use assumption $\emptyset \models e : \tau$ to get $(k+1,e) \in \epred{\tau}(\emptyset)$.
  If we use this with the assumed evaluation, $\irred(h',e')$, and the fact that $\hsat[k]{\emptyset}{\emptyset}$ (trivially true), then we get $W' \future \emptyset$ such that $\hsat[1]{h'}{W'}$ and $(1,e) \in \vpred{\tau}(W')$.
  This means that $e$ is indeed a value.
\end{proof}
It is now a simple matter to prove type safety.
\begin{theorem}[STLC with references is type safe]
  If $\emptyset \vdash e : \tau$, then $\safe(e)$
\end{theorem}
\begin{proof}
  Follows from Theorem~\ref{thm:ref-ftlr} and Lemma~\ref{lem:ref-adeq}.
\end{proof}
\subsection{Exercises}
\begin{enumerate}
%Verify that $k$-equality satisfies the two properties. (present third property and verify?)
\item Verify that the $k$-equality on $\UPred{A}$ satisfy the necessary properties. That is
  \begin{itemize}
  \item $\nequal[0]$ is the total relation, i.e.\ $\forall a,a' \in \UPred{A}\ldotp a \nequal[0] a'$.
  \item $\nequal[k+1] \subseteq \nequal$, i.e.\ $\forall a,a' \in \UPred{A}\ldotp a \nequal[k+1] a' \implies a \nequal a'$
  \end{itemize}
  A $k$-equality must also satisfy the following property omitted from the above presentation
  \[
    \forall a,a'\in \UPred{A}\ldotp (\forall k \in \nats a \nequal a') \implies a = a'
  \]
  That is, if two elements approximates each other for any index, then they should in fact be equal.

%Verify that future world relation is a partial order (describe what that means)
\item Verify that the future world relation is a partial order. That is for worlds $W$, $W'$ and $W''$ verify that it satisfies the following three properties
  \begin{itemize}
  \item Reflexivity: $W \future W$
  \item Antisymmetry: If $W' \future W$ and $W \future W'$, then $W = W'$
  \item Transitivity: If $W' \future W$ and $W'' \future W'$, then $W'' \future W$
  \end{itemize}
\item Prove the future world monotonicity lemma, i.e.\ Lemma~\ref{lem:world-monotonicity}.
\item Prove the downwards closure lemma, i.e.\ Lemma~\ref{lem:ref-downwards-closed}.
\item Prove the substitution lemma, i.e.\ Lemma~\ref{lem:ref-substitution}.
\item In the space where the recursive domain equation is solved, functions must preserve the added structure, i.e. step indices.
  Specifically, all functions must be non-expansive.
  A function is non-expansive if it preserves $k$-equalities, i.e. $f : X \rightarrow Y$ is non-expansive if the following holds true
  \[
    \forall x,x'\in X\ldotp \forall k \in N\ldotp x \nequal x' \implies f(x) \nequal f(x')
  \]
  Note that the first $k$-equality is defined for elements in $X$ and the second $k$-equality is defined for elements in $Y$.
  It is these $k$-equalities that must be preserved.
  The value interpretation is a function of the space in which the recursive domain equation is solved. Prove that this is the case.
% Aina, verify the same for all the other interpretation, i.e. the expression
% interpretation, the context typing interpretation, and the heap satisfaction
% (and any other function we have defined should I have forgotten about it).
\item Prove the remaining cases of Theorem~\ref{thm:ref-ftlr}, i.e.\ \textsc{T-Assign}, \textsc{T-Abs}, \textsc{T-App}, and \textsc{T-If}.
\end{enumerate}

%prove the mono lemma and dc lemma

%prove n.e. of a couple of definitions.

\subsection{Further Reading}
We refer the interested reader to the note \emph{Logical Relations and References}~\citep{Skorstengaard:references} for further reading.
It contains: important details that we omitted here for the sake of presentation, a more expressive language which allows for interesting examples but also calls for a different kind of world, and a logical relation rather than a predicate.
There is some overlap between this section and the note in question.

\section*{Acknowledgements} It is established practice for authors to accept responsibility for any and all mistakes in documents like this.
I, however, do not.
If you find anything amiss, please let me know so I can figure out who of the following are to blame: Amal Ahmed, Mathias Høier, Morten Krogh-Jespersen, Kent Grigo, and Kristoffer Just Andersen.

%TODO: Make appendix with all the correct definitions for easy lookup
\bibliography{ref}
\appendix
\section{Landin's Knot}
\label{app:landins-knot}
In this section, we show an implementation of
\begin{lstlisting}[escapeinside={@}{@},basicstyle=\footnotesize\ttfamily]
(let x = ref (\ y : int. y) in
    x := (\ n : int. (!x 0));
    !x) 0
\end{lstlisting}
in the language with references presented in Section~\ref{sec:references}.
We also demonstrate how this causes a recursion.

First we define the following gadgets:
\begin{align*}
  \dummy & = \lambda \; x : int \ldotp x\\
  \return & = \lambda y : \tarrow{\int}{\int} \ldotp \Ederef{x} \\
  \knot & = \Eassign{x}{(\lambda n : \int\ldotp (\Ederef{x} \; 0))}
\end{align*}
The implementation of the program is
\[
((\lambda x : \Tref{(\tau)}\ldotp \return\; \knot) \; (\Ealloc{ \dummy})) \; 0
\]
and it does indeed type check which the type derivation tree on the next page demonstrates:

\begin{landscape}
Let $\tau = \tarrow{\int}{\int}$
\begin{mathpar}
\inferrule*{
\inferrule*{
    \inferrule*{ }
             { x : \Tref{\tau} ; \mtenv \vdash x : \Tref{\tau} }
\\  
\inferrule*{
  \inferrule*{ 
    \inferrule*{
      \inferrule*{  }
                 { x : \Tref{\tau}, n : \int ; \mtenv \vdash x : \Tref{\tau} }}
               { x : \Tref{\tau}, n : \int ; \mtenv \vdash \Ederef{x} : \tau }
    \\
    \inferrule*{ }
               { x : \Tref{\tau}, n : \int ; \mtenv \vdash  0 : \int }
}
             { x : \Tref{\tau}, n : \int ; \mtenv \vdash \Ederef{x} \; 0 : \int }
}
{ x : \Tref{\tau} ; \mtenv \vdash \lambda n : \int\ldotp (\Ederef{x} \; 0) : \tau }
}
{ x : \Tref{\tau} ; \mtenv \vdash \Eassign{x}{(\lambda n : \int\ldotp (\Ederef{x} \; 0))} : \tau }}
{ x : \Tref{\tau} ; \mtenv \vdash \knot : \tau }
\end{mathpar}

\begin{mathpar}
  \inferrule* {
    \inferrule*{
      \inferrule*{
        \inferrule*{
          \inferrule* {
            \inferrule*{ }
                       {x : \Tref{\tau}, y : \tau ; \mtenv \vdash x : \Tref{\tau}} }
                      { x : \Tref{\tau}, y : \tau ; \mtenv \vdash \Ederef{x} : \tau }
        }
                    {x : \Tref{\tau} ; \mtenv \vdash \return : \tarrow{\tau}{\tau}}
\\
\inferrule*{
\vdots
}
          { x : \Tref{\tau} ; \mtenv \vdash \knot : \tau}
 }
                   { x : \Tref{\tau} ; \mtenv \vdash \return\; \knot : \tau }
 }
               {\mtenv ; \mtenv\vdash (\lambda x : \Tref{\tau}\ldotp\return\; \knot) : \tarrow{\Tref{\tau}}{\tau} }
    \\
    \inferrule*{ 
      \inferrule*{
        \inferrule*{ }
      { x : \int ; \mtenv \vdash x : \int }}
      { \mtenv ; \mtenv \vdash \dummy : \tarrow{\int}{\int} }
      }
  { \mtenv ; \mtenv \vdash \Ealloc{ \dummy} : \Tref{\tau} }
}
{\mtenv ; \mtenv \vdash (\lambda x : \Tref{(\tau)}\ldotp \return\; \knot) \;
  (\Ealloc{ \dummy}) : \tau
 }
\end{mathpar}
\begin{mathpar}
  \inferrule* {
    \inferrule*{ \vdots }
    { \mtenv ; \mtenv \vdash (\lambda x : \Tref{(\tau)}\ldotp \return\; \knot) \;
    (\Ealloc{ \dummy}) : \tau }
    \\
    \inferrule*{ }
    { \mtenv ; \mtenv \vdash 0 : \int }
  }
  {\mtenv ; \mtenv \vdash ((\lambda x : \Tref{(\tau)}\ldotp \return\; \knot) \;
    (\Ealloc{ \dummy})) \; 0 : \int}
\end{mathpar}
\end{landscape}
In order to see how the expression diverges, we first consider the left side of the application in which the knot is prepared:
\[
  \arraycolsep=0pt
  \begin{array}{rcl}
    \multicolumn{3}{l}{\tuple{h,(\lambda x : \Tref{(\tau)}\ldotp \return\; \knot) \; (\Ealloc{ \dummy})}} \\
    \phantom{quad}&\evalto{} & \tuple{h[l\mapsto\dummy],(\lambda x : \Tref{(\tau)}\ldotp \return\; \knot) \; l} \\
     &\evalto{} & \tuple{h[l\mapsto\dummy], \subst{(\return\; \knot)}{l}{x} } \\
     &\equiv{} & \tuple{h[l\mapsto\dummy], \subst{\return}{l}{x} \; (\Eassign{l}{(\lambda n : \int\ldotp (\Ederef{l} \; 0))}) } \\
     &\evalto{} & \tuple{h[l\mapsto \lambda n : \int\ldotp (\Ederef{l} \; 0)], \subst{\return}{l}{x} \; (\lambda n : \int\ldotp (\Ederef{l} \; 0)) } \\
     &\evalto{} & \tuple{h[l\mapsto \lambda n : \int\ldotp (\Ederef{l} \; 0)], \Ederef{l} } \\
     &\evalto{} & \tuple{h[l\mapsto \lambda n : \int\ldotp (\Ederef{l} \; 0)], \lambda n : \int\ldotp (\Ederef{l} \; 0) }
  \end{array}
\]
With the not prepared, $0$ is applied to activate it
\[
  \arraycolsep=0pt
  \begin{array}{rll}
    \multicolumn{3}{l}{\tuple{h,((\lambda x : \Tref{(\tau)}\ldotp \return\; \knot) \; (\Ealloc{ \dummy})) \; 0}} \\
    \phantom{quad}&\evaltos{} & \tuple{h[l\mapsto \lambda n : \int\ldotp (\Ederef{l} \; 0)], (\lambda n : \int\ldotp (\Ederef{l} \; 0)) \; 0 } \\
                  &\evalto{} & \tuple{h[l\mapsto \lambda n : \int\ldotp (\Ederef{l} \; 0)], (\Ederef{l} \; 0) }\\
                  &\evalto{} & \tuple{h[l\mapsto \lambda n : \int\ldotp (\Ederef{l} \; 0)], ((\lambda n : \int\ldotp (\Ederef{l} \; 0)) \; 0 } \\
                  &\evalto{} & \tuple{h[l\mapsto \lambda n : \int\ldotp (\Ederef{l} \; 0)], (\Ederef{l} \; 0) }\\
                  &\evalto{} & \tuple{h[l\mapsto \lambda n : \int\ldotp (\Ederef{l} \; 0)], ((\lambda n : \int\ldotp (\Ederef{l} \; 0)) \; 0 } \\
                  &\evalto{} & \tuple{h[l\mapsto \lambda n : \int\ldotp (\Ederef{l} \; 0)], (\Ederef{l} \; 0) }\\
                  &\evalto{} & \tuple{h[l\mapsto \lambda n : \int\ldotp (\Ederef{l} \; 0)], ((\lambda n : \int\ldotp (\Ederef{l} \; 0)) \; 0 } \\
                  &\evalto{} & \dots
  \end{array}
\]
the execution keeps alternating between two configurations, so it diverges.

\end{document}